\documentclass[superscriptaddress,twocolumn,floatfix,aps,10pt]{revtex4-1}
\usepackage{adjustbox,amsmath,amsfonts,amsthm,appendix,dsfont}
\usepackage[normalem]{ulem}
\usepackage{bbm}
\usepackage{booktabs}
\usepackage{physics}
\usepackage{float}
\usepackage{graphicx,wasysym}
 \usepackage{hyperref}
 \hypersetup{
     colorlinks=true,
     linkcolor=blue,
     filecolor=blue,
     citecolor = blue,      
     urlcolor=cyan,
     }
\usepackage{mathtools}
\usepackage{quantikz}
\usepackage{tikz}
\usepackage{xcolor}
\usepackage[ruled,vlined]{algorithm2e}
\usepackage{threeparttable}
\usepackage {pgfplots}
\pgfplotsset{compat=1.8}
\usepackage{lineno}
\usepackage{xparse}

\usepackage[caption=false]{subfig}

\DeclareMathOperator*{\argmin}{arg\,min}
\newtheorem{theorem}{Theorem}
\newtheorem*{theorem*}{Theorem}
\newtheorem{definition}[theorem]{Definition}
\newtheorem*{definition*}{Definition}
\newtheorem{lemma}[theorem]{Lemma}
\newtheorem{corollary}[theorem]{Corollary}

\renewcommand{\Vec}[1]{\mathbf{#1}}

\usepackage{xifthen}
\newcommand{\PP}[1][]{
  \ifthenelse{\isempty{#1}}
    {\mathbbm{P}}
    {\mathbbm{P}\left[#1\right]}
}
\newcommand{\EE}[1][]{
  \ifthenelse{\isempty{#1}}
    {\mathbbm{E}}
    {\mathbbm{E}\left[#1\right]}
}

\begin{document}
\title{Statistics-Informed Parameterized Quantum Circuit via Maximum Entropy Principle for Data Science and Finance}

\author{Xi-Ning Zhuang}
\thanks{These authors contributed equally to this work.}
\affiliation{CAS Key Laboratory of Quantum Information, University of Science and Technology of China, Hefei, 230026, China}
\author{Zhao-Yun Chen}
\thanks{These authors contributed equally to this work.}
\author{Cheng Xue}
\affiliation{Institute of Artificial Intelligence, Hefei Comprehensive National Science Center}
\author{Xiao-Fan Xu}
\affiliation{CAS Key Laboratory of Quantum Information, University of Science and Technology of China, Hefei, 230026, China}
\author{Chao Wang}
\affiliation{Origin Quantum Computing,  Hefei,  China}
\author{Huan-Yu~Liu}
\author{Tai-Ping Sun}
\affiliation{CAS Key Laboratory of Quantum Information, University of Science and Technology of China,  Hefei,  230026,  China}
\author{Yun-Jie Wang}
\affiliation{Institute of Advanced Technology, University of Science and Technology of China, Hefei, Anhui, 230031, China}
\author{Yu-Chun Wu}
\email{wuyuchun@ustc.edu.cn}
\author{Guo-Ping Guo}
\email{gpguo@ustc.edu.cn}
\affiliation{CAS Key Laboratory of Quantum Information, University of Science and Technology of China, Hefei, 230026, China}

\begin{abstract}
Quantum machine learning has demonstrated significant potential in solving practical problems, particularly in statistics-focused areas such as data science and finance. However, challenges remain in preparing and learning statistical models on a quantum processor due to issues with trainability and interpretability. In this letter, we utilize the maximum entropy principle to design a statistics-informed parameterized quantum circuit (SI-PQC) for efficiently preparing and training of quantum computational statistical models, including arbitrary distributions and their weighted mixtures. The SI-PQC features a static structure with trainable parameters, enabling in-depth optimized circuit compilation, exponential reductions in resource and time consumption, and improved trainability and interpretability for learning quantum states and classical model parameters simultaneously. As an efficient subroutine for preparing and learning in various quantum algorithms, the SI-PQC addresses the input bottleneck and facilitates the injection of prior knowledge.
\end{abstract}

\maketitle

\textit{Introduction} -- In data science, statistical models are simplified representations of real-world data using stochastic mathematical tools. These models are typically represented as preset distributions based on prior knowledge. Model calibration, a central task in this field, involves estimating parameters of these distributions using observed data~\cite{hastie:2009elements,casella:2024statistical,chambers:2017statistical}.
However, model calibration can often be challenging: analytical methods may fail when maximum-likelihood functions are too complex or when closed-form solutions do not exist~\cite{hastie:2009elements,casella:2024statistical,do:2008what, lee:2019review}. Classical numerical methods based on search or simulation, such as the expectation-maximization algorithm and Monte Carlo simulation, become prohibitively expensive when dealing with large datasets~\cite{hastie:2009elements,geof:2012computational}.

Quantum machine learning (QML) is a promising approach for handling large-scale data~\cite{huang:2021power, liu:2021quantum, huang:2021information, Cerezo:2022challenges, Yazhen:2022quantum}. 
However, whether quantum computing can provide a significant advantage in model calibration remains an emerging question. 
Existing model-free methods, such as Quantum Generative Adversarial Networks (QGANs)~\cite{zoufal:2019quantum}, can learn the distribution but cannot extract the underlying parameters that are essential for statistical interpretability-concerned tasks like financial derivatives pricing and gene clustering~\cite{do:2008what,lee:2019review, bracke:2019machine, vellido:2020importance}.
Model-based methods can prepare a quantum distribution state with fixed parameters, but these methods are hard to train with parametric configurations~\cite{Grover:2002creating, holmes:2020efficient, Iaconis:2024quantum, zylberman:2023efficient, rattew:2022preparing, conde:2023efficient, guo:2021nonlinear, rattew:2023non}. 
Furthermore, additional costs are incurred when using these methods to prepare a mixture of distributions with varying and latent parameters, which is an essential task in statistics-related problems and machine learning applications~\cite{lee:2019review,warner:2015stochastic,lee:2012application,hideyuki:2020quantum,teicher:1960on,douglas:2009gaussian,ian:2008portfolio}. 

\begin{figure*}
    \centering
    \begin{tikzpicture}
        \node[] at (-4cm, 1.35cm) {(a)};
        \begin{scope}[x=1cm,y=1cm,shift={(-3cm, 0cm)}]
            \draw[step=0.25, very thin, color=gray!50] (-1,-1) grid (1,1);
            \draw[color=violet!50, fill=violet!50] (-0.95, -0.75) rectangle (-0.8, -0.5) {};
            \draw[color=violet!50, fill=violet!50] (-0.7, -0.75) rectangle (-0.55, -0.25) {};
            \draw[color=violet!50, fill=violet!50] (-0.45, -0.75) rectangle (-0.3, 0.25) {};
            \draw[color=violet!50, fill=violet!50] (-0.2, -0.75) rectangle (-0.05, 0.15) {};
            \draw[color=violet!50, fill=violet!50] (0.05, -0.75) rectangle (0.2, 0.6) {};
            \draw[color=violet!50, fill=violet!50] (0.3, -0.75) rectangle (0.45, -0.15) {};
            \draw[color=violet!50, fill=violet!50] (0.55, -0.75) rectangle (0.7, -0.05) {};
            \draw[color=violet!50, fill=violet!50] (0.8, -0.75) rectangle (0.95, -0.35) {};
            \draw[black,->] (-1,-0.75) -- (1,-0.75);
            \draw[black,->] (0,-1) -- (0, 1);
        \end{scope}
        \node[scale=0.8] at (-3cm, -1.5cm) {{Empirical}};
        \node[scale=0.8] at (-3cm, -1.75cm) {{Distribution}};
        \draw[black,-latex] (-2cm, 0cm) -- (-1.15cm, 0.8cm);
        \draw[black,-latex] (-2cm, 0cm) -- (-1.15cm, 0cm);
        \draw[black,-latex] (-2cm, 0cm) -- (-1.15cm, -0.8cm);
        \draw[dashed, thick, color=cyan] (-1.1cm, -1.27cm) rectangle (-0.1cm, 1.27cm) {};
        \begin{scope}[x=0.5cm,y=0.5cm,shift={(-0.6cm, 0.8cm)},scale=0.75]
            \draw[domain=-1:1, smooth, very thick, variable=\x, orange]  plot ({\x}, {\x});
            \draw[step=0.5, very thin, color=gray!30] (-1,-1) grid (1,1);
            \node[]{$f_1$};
        \end{scope}
        \begin{scope}[x=0.5cm,y=0.5cm,shift={(-0.6cm, 0cm)},scale=0.75]
            \draw[domain=-1:1, smooth, very thick, variable=\x, orange]  plot ({\x}, {\x*\x});
            \draw[step=0.5, very thin, color=gray!30] (-1,-1) grid (1,1);
            \node[]{$f_k$};
        \end{scope}
        \begin{scope}[x=0.5cm,y=0.5cm,shift={(-0.6cm,-0.8cm)},scale=0.75]
            \draw[domain=-1:1, smooth, very thick, variable=\x, orange]  plot ({\x}, {\x*\x*\x});
            \draw[step=0.5, very thin, color=gray!30] (-1,-1) grid (1,1);
            \node[]{$f_M$};
        \end{scope}
        \node[scale=0.8] at (-0.6cm, -1.5cm) {{Observables as}};
        \node[scale=0.8] at (-0.6cm, -1.75cm) {{Constraints}};
        \draw[dashed, thick, color=orange](2.75cm, 0cm) circle(0.85cm) node[color=black]{};
        \draw[black] (-0.15cm, 0.8cm) -- (2.75cm, 0.8cm);
        \draw[black] (-0.15cm, 0cm) -- (2.75cm, 0cm);
        \draw[black] (-0.15cm, -0.8cm) -- (2.75cm, -0.8cm);
        \draw[black,-latex] (2.75cm, -0.8cm) -- (4.15cm, -0.2cm);
        \draw[black,-latex] (2.75cm, 0cm) -- (4.15cm, 0cm);
        \draw[black,-latex] (2.75cm, 0.8cm) -- (4.15cm, 0.2cm);
        
        \draw[dashed, thick, color=green] (0.85cm, -1.27cm) rectangle (1.65cm, 1.27cm) {};
        \draw[fill=yellow!50](1.25cm, 0.8cm) circle(0.27cm) node[color=black]{$\lambda_1$};
        \node[rotate=90] at (0.45cm, 0.4cm) {{\ldots}};
        \node[rotate=90] at (0.45cm, -0.4cm) {{\ldots}};
        \draw[fill=red!50](1.25cm, 0cm) circle(0.27cm) node[color=black]{$\lambda_k$};
        \draw[fill=blue!30](1.25cm, -0.8cm) circle(0.27cm) node[color=black]{$\lambda_M$};
        \node[scale=0.8] at (1.25cm, -1.5cm) {{Linear}};
        \node[scale=0.8] at (1.25cm, -1.75cm) {{Combination}};
        \node[scale=0.18] at (2cm, -0.5cm) {
        \begin{axis}[thin]
            \addplot3[domain=-5:5,mesh] {(50-x*x-y*y)};
        \end{axis}
        };
        \node[scale=0.8] at (2.75cm, -1.5cm) {{Maximize}};
        \node[scale=0.8] at (2.75cm, -1.75cm) {{Entropy}};
        \begin{scope}[x=1cm,y=1cm,shift={(5.25cm, 0cm)}]
            \draw[step=0.25, very thin, color=gray!50] (-1,-1) grid (1,1);
            \draw[color=violet!50, fill=violet!50] (-0.95, -0.75) rectangle (-0.8, -0.5) {};
            \draw[color=violet!50, fill=violet!50] (-0.7, -0.75) rectangle (-0.55, -0.25) {};
            \draw[color=violet!50, fill=violet!50] (-0.45, -0.75) rectangle (-0.3, 0.25) {};
            \draw[color=violet!50, fill=violet!50] (-0.2, -0.75) rectangle (-0.05, 0.15) {};
            \draw[color=violet!50, fill=violet!50] (0.05, -0.75) rectangle (0.2, 0.6) {};
            \draw[color=violet!50, fill=violet!50] (0.3, -0.75) rectangle (0.45, -0.15) {};
            \draw[color=violet!50, fill=violet!50] (0.55, -0.75) rectangle (0.7, -0.05) {};
            \draw[color=violet!50, fill=violet!50] (0.8, -0.75) rectangle (0.95, -0.35) {};
            \draw[black!50,->] (-1,-0.75) -- (1,-0.75);
            \draw[black!50,->] (0,-1) -- (0, 0.75);
            \draw[domain=-1:1, smooth, very thick, variable=\x, orange]  plot ({\x}, {exp(-2*\x*\x)-0.725});
            \node[very thick] at (0,1.15) {$p(x)=e^{\sum_{k=1}^M\lambda_kf_k(x)}$};
        \end{scope}
        \node[scale=0.8] at (5.25cm, -1.5cm) {{Maximal Entropy}};
        \node[scale=0.8] at (5.25cm, -1.75cm) {{Distribution}};
        \node[] at (7.5cm, 1.35cm) {(b)};
        \draw [color=gray!20,fill=gray!20] plot [smooth cycle] coordinates {(7.5cm,0.6cm) (8.5cm,1.1cm) (10cm,0.1cm) (9.5cm,-1.4cm) (8cm,-0.4cm)};
        \draw [color=cyan!30,fill=cyan!30] plot [smooth cycle] coordinates {(8cm,0.1cm) (8.5cm,0.6cm) (9.5cm,-0.4cm) (9cm,-0.4cm)};
        \draw[thick,color=red,-latex] plot [smooth] coordinates{(8.2cm, 0.1cm) (8.5cm, 0.1cm) (9cm, -0.2cm) (9.3cm, -0.18cm)};
        \draw[thick,color=teal,dotted,-latex] plot [smooth] coordinates{(8.2cm, 0.1cm) (8.5cm, -0.4cm) (9cm, -0.9cm) (9.4cm, -0.7cm) (9.3cm, -0.18cm)};
        \node[] at (9.5cm, 1.3cm) {Target space};
        \draw[] (9.5cm, 1.1cm) -- (8.5cm, 0.4cm);
        \node[] at (8.7cm, -1.8cm) {Hilbert space};
        \draw[] (8.5cm, -1.6cm) -- (9.3cm, -1.1cm);
    \end{tikzpicture}
    \caption{Schematic depiction of Statistics-Informed Parameterized Quantum Circuit (SI-PQC) via maximum entropy principle (MEP). 
    (a) MEP promises a unified formula to describe empirical data as linearly combined constraints from observables with maximized entropy.
    (b) SI-PQC achieves a universal expresiveness in the target space and the output state stays therein during optimization procedure (solid red arrow). While problem agnostic parameterized quantum circuit output state can often violate from the target space when varying circuit parameters (dotted green arrow).}
    \label{fig:SI-PQC}
\end{figure*}
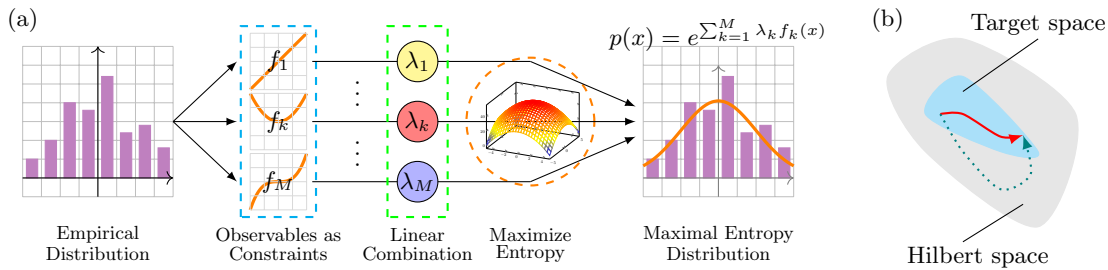
Here, we propose the statistics-informed parameterized quantum circuit (SI-PQC) inspired by maximum entropy principle, as depicted in Fig.~\ref{fig:SI-PQC}(a), to generate and train statistics models with fixed circuit structure and tunable parameters explicitly encoded on rotation angles.
For preparing tasks, the costly pre-computing or pre-training procedure is no longer needed~\cite{holmes:2020efficient,Iaconis:2024quantum,gilyen:2018quantum, yulong:2021efficient}, and hence any in-depth circuit compilation optimizing is permitted.
Further, an exponential reduction of time and quantum resource consumption is achieved when preparing the mixture of numerous statistics models.
For learning tasks, SI-PQC achieves universal expressiveness in the subspace of interest with an optimal number of parameters. It improves the trainability, avoiding an unexpectedly large output space of problem-agnostic parameterized quantum circuit (PQC) ansatz~\cite{Cerezo:2022challenges, mcclean:2018barren, martin:2024a, zoe:2022connecting}, as illustrated in Fig.~\ref{fig:SI-PQC}(b).

The main idea of SI-PQC is cut in from the physics viewpoint of the maximum entropy principle that
an arbitrary distribution can be regarded as a maximum entropy distribution (MED), that is   
\begin{equation}\label{eq:maximum_entropy_principle}
    p(x) = \frac{1}{\mathcal{N}}\exp{\sum_{k=1}^M \lambda_k f_k(x)}.
\end{equation}
Herein $\mathcal{N}$ is the normalized constant, constraints $f_k$ are $M$ different observables of empirical data where $\mathbb{E}[f_k(x)]=a_k (1\leq k\leq M)$, and $\lambda_k$ are Lagrange multipliers which determine the distribution, as elaborated in Eq.~\eqref{eq:maximum_entropy_principle}.
The generalization of MED is two-fold. First, any distribution is proven to be a MED.
Second, MED is the best-unbiased guess with given observables from data when the explicit formula is not accessible~\cite{conrad:2004probability, banavar:2010applications, park:2009maximum}.

\textit{Maximum entropy distribution loader} --
Eq.~\eqref{eq:maximum_entropy_principle} is exponentiation of linearly combined constraints, which enables us to propose the following SI-PQC named \textit{maximum entropy distribution loader} (MEDL) in three steps, as shown in Fig.~\ref{fig:medl}(a).

Firstly, we prepare the normalized distinct constraint functions $\vec{f}_k = \frac{1}{2\lVert f_k\rVert_{\mathrm{max}}}f_k$ through a sequence of fixed constraint layers of unitaries $U_k$ as the quantum singular value transformation of the diagonal block-encoding of $(0, \frac{1}{2^n-1}, \frac{2}{2^n-1}, ..., 1)$.
The constraints $f_k$ are usually elementary functions with parameters fixed to be $1$, which simplifies the preparation by eliminating the need for the costly and numerically unstable phase factor evaluation~\cite{gilyen:2018quantum, yulong:2021efficient}. The constraints of common distributions are listed as a table in SM~\cite{SM}.

Secondly, we inject the varying statistics information $\vec{\lambda} = (\lambda_k)_{1\leq k\leq M}$ into the parameterized layer as the rotation angles
\begin{equation}\label{eq:lcc_rotation_angle}    \vec{\theta}=\left(2\arctan\sqrt{\frac{\sum_{j=k+1}^M2\lambda_j \lVert f_j\rVert_{\mathrm{max}}}{2\lambda_k \lVert f_k\rVert_{\mathrm{max}}}}\right)_{1\leq k\leq M}
\end{equation}
to encode the sparse linear combination of constraints by linear combination of unitaries (LCU)~\cite{childs:2012hamiltonian}. This step creates a block-encoding of 
\begin{equation}
    f = \sum_{k=1}^{M} \lambda_kf_k(x).
\end{equation}
Note that Eq.~\eqref{eq:lcc_rotation_angle} gives an optimal parameter space dimension since at least $M$ free parameters are needed.

Finally, we apply the imaginary time evolution based on the block-encoded $f$ by the exponentiation layer. In specific, we apply an additional quantum singular value transformation for $d=\log{1/\epsilon}$ Taylor expansion of $e^x$, within the coefficients is obviously fixed and the phase angles $\phi_j (1\leq j \leq d)$ can be pre-computed~\cite{gilyen:2018quantum}.

\begin{figure*}[ht]
    \centering
    \begin{tikzpicture}
        \node[scale=0.6] at (0, 0) {\begin{quantikz}[font=\large]
                \lstick[3,brackets=right,label style={ xshift=-0.5cm, yshift=1.75cm, rotate=90}]{Parameters Injection}&&\gate[3, style={fill=green!10}, label style={rotate=90}]{P_R(\textcolor{red}{\theta})}
                \gategroup[5,steps=6,background,style={rounded corners, inner xsep=1pt, inner ysep=-1pt, draw=black!10, fill=black!5}, label style={label position=above,anchor=north,yshift=0.35cm}]{Linear Combination of Constraints}
                \gategroup[6,steps=7,background,style={dashed,rounded corners, inner xsep=5pt, inner ysep=15pt}, label style={label position=below,anchor=north,yshift=-0.2cm}]{Repeat $\log{(1/\epsilon\mathcal{F})}$ times for Imaginary Time Evolution}
                &\ctrl{3}&&&&\gate[3, style={fill=green!10}, label style={rotate=90}]{P_L(\textcolor{red}{\theta})^\dagger}&&\\
                &&&&\ctrl{2}&&&&&\\[1cm]
                & &&&&\lstick[1,brackets=none,label style={yshift=0.75cm, xshift=0cm,color=black, rotate=135}]{\ldots}\ \ldots\ &\ctrl{1}&&&\\[1cm]
                \lstick[1,brackets=none,label style={yshift=0.5cm, xshift=-0.5cm, rotate=90}]{{Data}}&\qwbundle{n}&&\gate[2, style={fill=cyan!10}]{U_{f_1}}&\gate[2, style={fill=cyan!10}]{U_{f_2}}&\ \ldots\ &\gate[2, style={fill=cyan!10}]{U_{f_M}}&  \arrow[arrows]{ll} &&  \arrow[arrows]{ll} \\[1cm]
                \lstick[1,brackets=none,label style={yshift=1cm, xshift=-0.5cm, rotate=90}]{{Constraint\\Anc.}}&\qwbundle{\log{n}}&&&&&&  \arrow[arrows]{ll} &\gate[2, style={fill=orange!20, text width=30pt}, label style={rotate=90}]{e^{i\textcolor{blue}{\phi_j}(2\Pi-I)}}&  \arrow[arrows]{ll}\\[1cm]
                \lstick[1,brackets=none,label style={yshift=0.5cm, xshift=-0.5cm, rotate=90}]{{ITE\\Anc.}}&\qwbundle{1}&&&&&&  \arrow[arrows]{ll} &&  \arrow[arrows]{ll}
            \end{quantikz}};
        \node at (9, 0) {\includegraphics[width=0.5\linewidth]{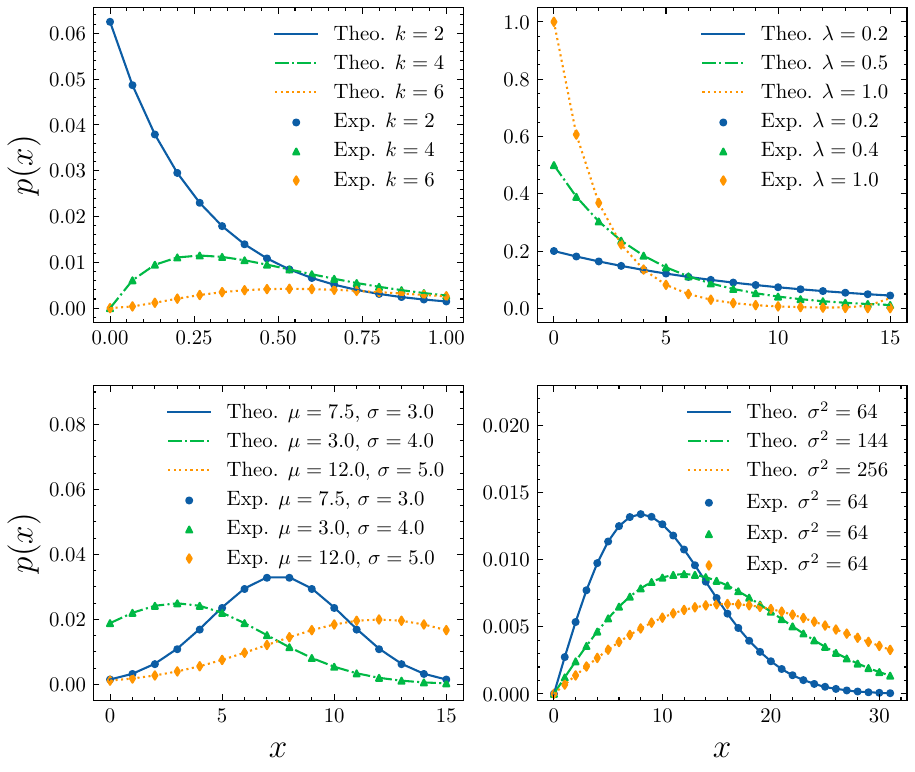}};
        \node at (-4, 3.5) {\text{(a)}};
        \node at (4.5, 3.5) {\text{(b)}};
        \node[scale=0.8] at (7.3, 3.85) {\text{Chi-squared}};
        \node at (9.27, 3.5) {\text{(c)}};
        \node[scale=0.8] at (11.6, 3.85) {\text{Exponential}};
        \node at (4.5, 0) {\text{(d)}};
        \node[scale=0.8] at (7.35, 0.15) {\text{Normal}};
        \node at (9.27, 0) {\text{(e)}};
        \node[scale=0.8] at (11.7, 0.15) {\text{Rayleigh}};
    \end{tikzpicture}
    \caption{{Maximum entropy distribution loader.}
    {(a)} The quantum circuit for MEDL.
    {(b-e)} The numerical results to prepare Chi-squared (upper left), exponential (upper right), normal (lower left) and Rayleigh (lower right) distributions.
    The dots represent the experimental results and lines are the theoretical distributions.}
    \label{fig:medl}
\end{figure*}
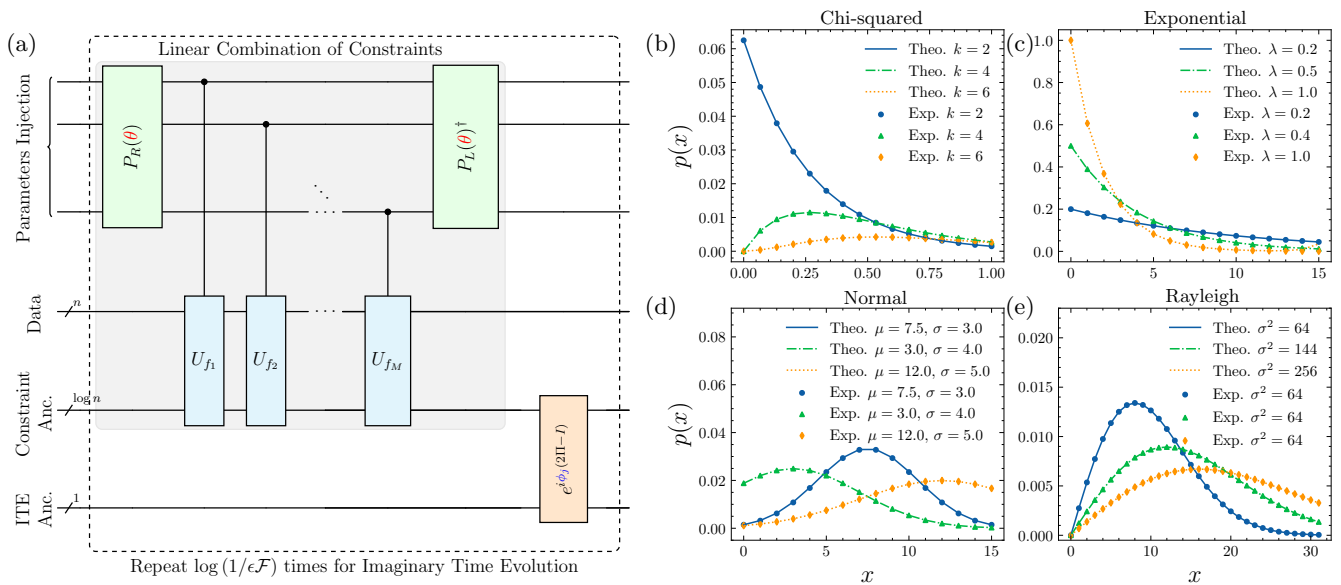

To summarize, we derive the following result:
\begin{theorem}[\textbf{Maximum Entropy Distribution Loader}]\label{thm:med-SI-PQC}
    Any statistical distribution can be viewed as a maximum entropy quantum state $\ket{p} = \sum_x\frac{1}{\mathcal{N}}\exp{\sum_{k=1}^M \lambda_k f_k(x)}\ket{x}$
    and can be prepared via a statistics-informed parameterized quantum circuit with a fixed structure and trainable parameters.
    The circuit depth is $\mathcal{O}\left(nMd_f\log\left(\frac{1}{\mathcal{F}\epsilon}\right)\right)$ with $\mathcal{O}(\log{n})$ ancillary qubits and success probability $\mathcal{O}(1/\mathcal{F})$, given the data space qubit number $n$, constraints number $M$, maximum constraint order $d_f$, filling rate $\mathcal{F}=\frac{\lVert p\rVert_{2}}{\lVert p\rVert_{\mathrm{max}}}$, and precision $\epsilon$.
\end{theorem}
\noindent Herein, the circuit depth scales linearly as the qubit number $n$ and hence can serve as an efficient amplitude-encoded state preparation subroutine for many QML and quantum-enhanced Monte-Carlo integration tasks~\cite{hideyuki:2020quantum,worner:2019quantum, Blank:2021quantum, Zhuang:2023quantum}.
The circuit depth also depends linearly on the free parameter numbers $M$ and the maximum constraint order $d_f$, which are usually a small constant, as detailed in SM~\cite{SM}.
The logarithmic dependency on the precision $1/\epsilon$ comes from the imperfect rotation angles.
The success probability dependency on the filling rate~\cite{rattew:2022preparing, rattew:2023non}, which is to characterize the amplitude-encoded preparation complexity of continuous function, is equal to the optimal analytic result up to known.

We build MEDL in the explicit quantum circuit and show output states in Fig.~\ref{fig:medl}(b-e). Here, we demonstrate exponential, Chi-squared, normal, and Rayleigh distributions with applications in financial time series, statistics, machine learning, and physics, respectively. For each distribution family, several distribution parameters are tested. All experiments match well with their theoretical results, representing the generality of the MEDL.

A key feature of MEDL is that
the varying model parameters $\vec{\lambda}$ are encoded explicitly into the rotation angles $\vec{\theta}$ in parameterized layers while the constraint layers and exponentiation layers remain invariant. 
This provides two main advantages. 
First, the circuit structure is static, enabling advanced circuit compilation, deep optimization, reuse, and high-precision computation of the phase factor, which is considered a primary error source in QSVT~\cite{yulong:2021efficient,rattew:2023non}.
Second, most layers in the circuit remain static for preparing distributions of the same family, allowing an efficient scheme for preparing a mixture of distributions.

\textit{Mixture of distributions} -- 
The mixture of distributions, which are widely used in machine learning, finance, astronomy, and robust statistics \cite{hideyuki:2020quantum,douglas:2009gaussian,ian:2008portfolio,lee:2012application, huber:2011robust}, can be efficiently prepared by SI-PQC by reusing components in MEDL.

Formally, given distribution of latent parameter space and visible distribution of data space to be $\mathcal{P}_\mathrm{para}$, respectively, we can prepare
\begin{equation}\label{eq:2}
    p(x) = \int \mathcal{P}_\mathrm{data}(x; \theta) \mathcal{P}_\mathrm{para}(\theta) d\theta
\end{equation}
by \textit{Weighted Distribution Mixer} (WDM) pictured in Fig.~\ref{fig:mixer}(a).

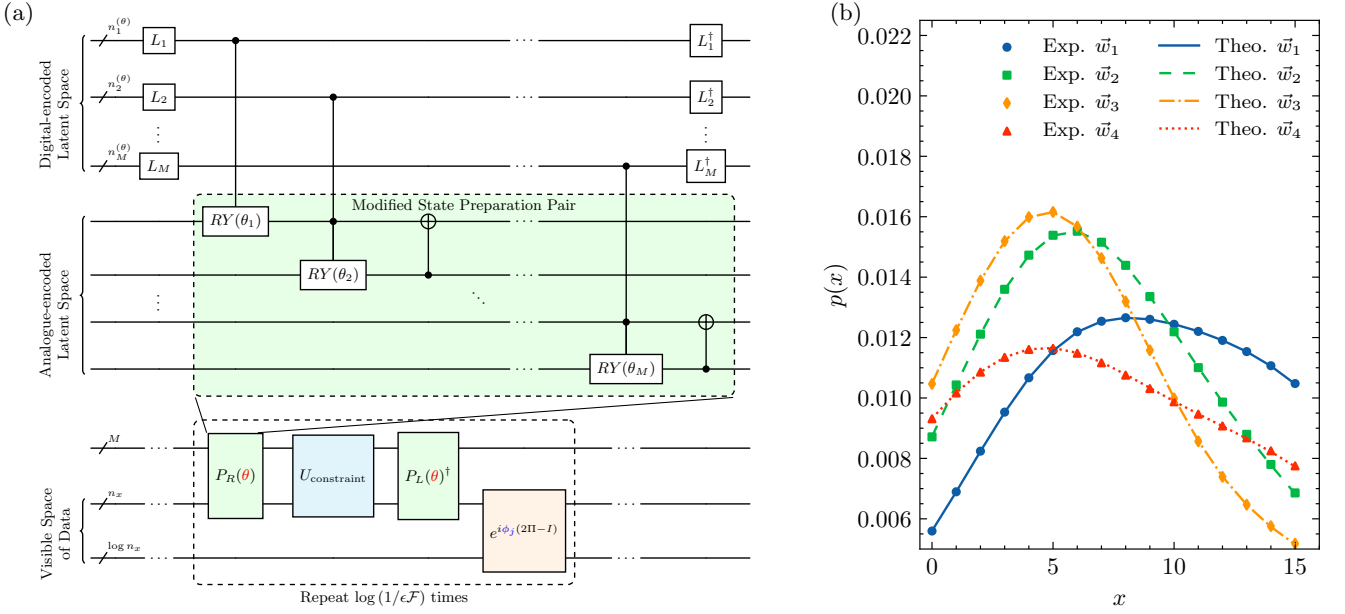
\begin{figure*}
    \centering
    \begin{tikzpicture}
        \node[scale=0.65] at (0, 0) {
            \begin{quantikz}
                \lstick[3,brackets=right,label style={rotate=90, yshift=0.5cm, xshift=1cm}]{{Digital-encoded\\ Latent Space}}&\qwbundle{n_1^{(\theta)}}
                 &\gate{L_1}&\ctrl{3}&&&\ \ldots\ &&\gate{L_1^\dagger}&\\ 
                &\qwbundle{n_2^{(\theta)}} &\gate{L_2}&&\ctrl{3}&&\ \ldots\ &&\gate{L_2^\dagger}&\\[0.25cm]
                &\qwbundle{n_M^{(\theta)}}\lstick[1,brackets=none,label style={yshift=1cm, xshift=0.9cm,color=black, rotate=90}]{\ldots} &\gate{L_M}&&&&\ \ldots\ &\ctrl{4}&\gate{L_M^\dagger}\lstick[1,brackets=none,label style={yshift=1cm, xshift=0.4cm,color=black, rotate=90}]{\ldots}&\\
                \lstick[4,brackets=right,label style={rotate=90, yshift=0.5cm, xshift=1cm}]{{Analogue-encoded\\ Latent Space}}&&&\gate{RY(\theta_1)}\gategroup[4,steps=6,background,style={dashed,rounded corners,fill=green!10, inner xsep=2pt}, label style={label position=above,anchor=north,yshift=-0.2cm}]{Modified State Preparation Pair}&\ctrl{1}&\targ{}&\ \ldots\ &&&\\
                &&&&\gate{RY(\theta_2)}&\ctrl{-1}&\ \ldots\ &&&\\
                &\lstick[1,brackets=none,label style={yshift=0.75cm, xshift=0.9cm,color=black, rotate=90}]{\ldots}&&&&\lstick[1,brackets=none,label style={yshift=0.75cm, xshift=0.8cm,color=black, rotate=135}]{\ldots}&\ \ldots\ &\ctrl{1}&\targ{}&\\
                &&&&&&\ \ldots\ &\gate{RY(\theta_M)}&\ctrl{-1}&\\[0.5cm]
                &\qwbundle{M}&\ \ldots\ &\gate[2, style={fill=green!10}]{P_R(\textcolor{red}{\theta})}
                \gategroup[3,steps=4,background,style={dashed,rounded corners, inner xsep=2pt}, label style={label position=below,anchor=north,yshift=-0.2cm}]{Repeat $\log{(1/\epsilon\mathcal{F})}$ times}
                &\gate[2, style={fill=cyan!10}]{U_{\mathrm{constraint}}}&\gate[2, style={fill=green!10}]{P_L(\textcolor{red}{\theta})^\dagger}&&\ \ldots\ &&\\
                \lstick[3,brackets=right,label style={rotate=90, yshift=0.5cm, xshift=1cm}]{{Visible Space\\ of Data}}&\qwbundle{n_x}&\ \ldots\ &&&&\gate[2, style={fill=orange!10}]{e^{i\textcolor{blue}{\phi_j}(2\Pi-I)}}&\ \ldots\ &&\\
                &\qwbundle{\log{n_x}}&\ \ldots\ &&&&&\ \ldots\ &&
            \end{quantikz}
        };
        \draw (-2.5,-1.57) -- (-2.65, -1.1);
        \draw (-1.82,-1.57) -- (4.5, -1.1);
        \node at (9, 0) {\includegraphics[width=0.38\linewidth]{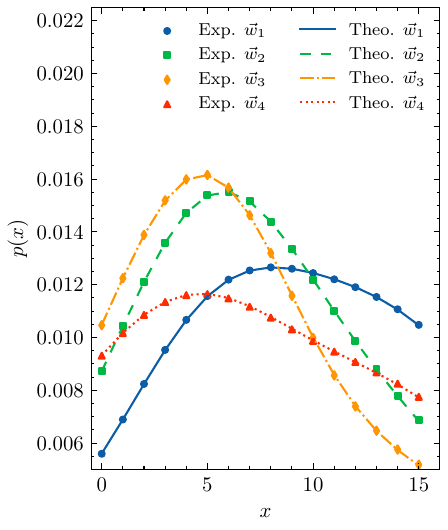}};
        \node at (-5, 4) {\text{(a)}};
        \node at (6, 4) {\text{(b)}};
    \end{tikzpicture}
    \caption{{Weighted distribution mixer.}
    {(a)} The quantum circuit for the Weighted Distribution Mixer.
    The state preparation pair $(P_L, P_R)$ in MEDL is substituted by an analogue-encoded space controlled by the digital-encoded latent space.
    {(b)} The numerical result to prepare a $4$-mixture of normal distributions with means $\Vec{\mu}=(1, 5, 10, 15)$, variances $\sigma^2=(30.25, 16, 36, 36)$, and weights $\vec{w}_1=(0.1, 0.2, 0.3, 0.4)$, $\vec{w}_2=(0.2, 0.3, 0.4, 0.1)$, $\vec{w}_3=(0.3, 0.4, 0.1, 0.2)$, $\vec{w}_4=(0.4, 0.1, 0.2, 0.3)$.}    
    \label{fig:mixer}
\end{figure*}
The procedure of WDM is as follows. First, we prepare the quantum superposition of $N_\theta$ possible model configurations $\theta$ on the latent space of $n_\theta=\log{N_\theta}$ qubits with $d_\theta$ circuit depth.
Then a digital-analog conversion is utilized to implement the parameterized layer~\cite{mitarai:2019quantum}. A detailed description is shown in SM~\cite{SM}.
In summary, we have:
\begin{theorem}[\textbf{Weighted Distribution Mixer of a Parametric Family}]\label{thm:wdm-SI-PQC}
    An arbitrary distribution mixture of a parametric family $f(x; \theta)$ can be modeled by an entangled quantum system $\frac{1}{\mathcal{N}}\sum_\theta\sum_x{\mathcal{P}_\mathrm{para}(\theta)}\mathcal{P}_\mathrm{data}(x; \theta)\ket{\theta}\ket{x}$ and can be prepared via a statistics-informed parameterized quantum circuit with fixed structure and trainable parameters. The total circuit depth is $\mathcal{O}(({n_xMd_f}+ n_{\theta}+d_\theta)\log{\frac{1}{\epsilon\mathcal{F}}})$ given an $M$-dimension latent parameter space of ${n_{\theta}}$ qubits and depth $d_\theta$, a visible data space of ${n_{x}}$ qubits, $M$ corresponding constraints with maximum order $d_f$, and filling rate $\mathcal{F}$.
\end{theorem}
Note that instead of requiring $N_\theta$ repetitions of the whole circuit, WDM only necessitates an additional depth of $\mathcal{O}((\log{N_\theta}+d_\theta)\log{\frac{1}{\epsilon\mathcal{F}}})$ to implement the $N_\theta$-mixture of
parametric distributions and to model the latent distribution simultaneously. Consequently, an exponentially growing circuit depth is avoided due to the static structure and invariant output space of SI-PQC.

We also build the quantum circuit for WDM. In the test, we prepare the mixture of four normal distributions with different weights $w_k$, means $\mu_k$ and variances $\sigma_k^2$, that is
\begin{equation}
    p(x)=\sum_k w_k \frac{1}{\sqrt{2\pi}\sigma_k}e^{-\frac{(x-\mu_k)^2}{2\sigma_k^2}}.
\end{equation}
Results show a near-perfect match to the theory, as shown in Fig.~\ref{fig:mixer}(b).
Our methods can be naturally generalized to  mixtures of other distributions. Additional numerical results are shown in SM~\cite{SM}.

\textit{Data-driven model calibration} --
In practice, the SI-PQC method, including MEDL and WDM, can create a superposition of numerous statistics models with all possible parameter configurations, viewed as the quantum brute force of statistics modeling, discussed in the context of stochastic process in Ref.~\cite{Blank:2021quantum}. 
To showcase the applicability in statistics learning, we apply a quantum-classical hybrid optimization procedure to minimize the calibration metric between the observed and predicted data
$\argmin_{\theta} f(Y_{\mathrm{obs}}, Y_{\mathrm{pre}}(x|\theta))$.
The objective function is given by
$\braket{P_{\mathrm{pre}}}{P_{\mathrm{obs}}} = \bra{0}U_{\mathrm{SI-PQC}}^\dagger O_\mathrm{obs}\ket{0}$, and can be efficiently computed by SI-PQC.

\begin{figure}
    \begin{tikzpicture}
        \node[]{
    \includegraphics[width=\linewidth]{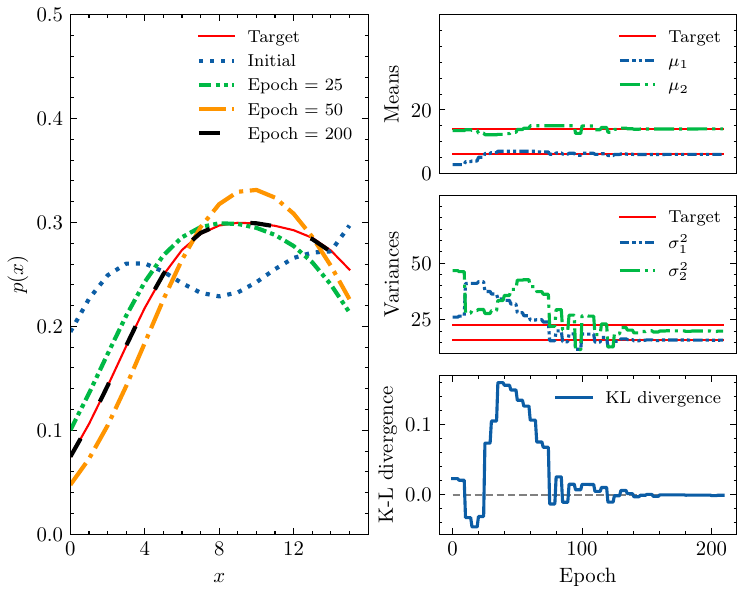}};
    \node[] at(-4.25, 3.2) {(a)};
    \node[] at(0.5, 3.2) {(b)};
    \node[] at(0.5, 1) {(c)};
    \node[] at(0.5, -0.9) {(d)};
    \end{tikzpicture}
    \centering
    \caption{
    {Learning Gaussian Mixture Model based on SI-PQC.}
    {(a)}
    Quantum states in the learning procedure.
    {(b)}
    Estimated means extracted from the angles in the training procedure.
    {(c)}
    Estimated variances extracted from the angles in the training procedure.
    {(d)}
    The K-L divergence to characterize the distance from the target state.
    }
    \label{fig:learning}
\end{figure}
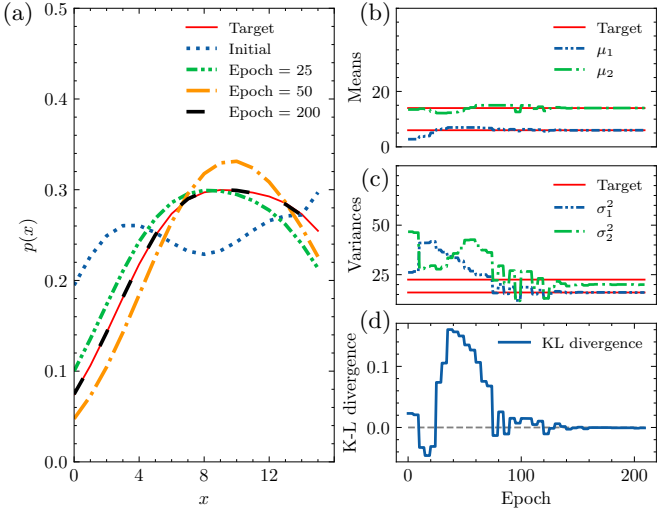

In our numerical experiment, we demonstrate the learning of the Gaussian Mixture Model (GMM), a statistics model often used in machine learning and finance~\cite{hideyuki:2020quantum,douglas:2009gaussian,ian:2008portfolio}.
The classical GMM estimation procedure is usually costly and the quantum method is proposed assuming a quantum access of GMM~\cite{hideyuki:2020quantum}.
We use the fidelity with the black-box oracle as the loss function.
The target mixture of two Gaussian distributions of means $\mu_1=6, \mu_2=14$, variances $\sigma_1^2 = 16, \sigma_2^2 = 22.5$, and weights $w_1 = 0.45, w_2 = 0.55$ is depicted by the red dashed line in Fig.~\ref{fig:learning}(a).
We started from a random initialized state colored in blue and learned the final state colored in black.
Moreover, SI-PQC enables us to learn the model parameters as well as the quantum state simultaneously.
Fig.~\ref{fig:learning}(b-c) show the estimated model parameters converging to the target means and variances, respectively.
This illustrates the better interpretability beyond other variational-based learning protocols.
The experimental results agree very well with the target distribution, characterized by the K-L divergence and fidelity relative to the target distribution, depicted in Fig.~\ref{fig:learning}(d). Note that the training protocol of SI-PQC is not limited to any specific distribution, where we show the training of exponential mixture model in SM~\cite{SM}.

The potential quantum advantage is twofold:
Firstly, suppose that one is given access to quantum data, such as the intermediate state from another quantum algorithm that needs further statistical characterization or evaluation.
Instead of a costly quantum state tomography procedure, SI-PQC-based model calibration can save the consumption of quantum states and time.
Secondly, suppose that one is given classical data from the realistic world, quantum model calibration procedure still possesses a potential advantage where one can evaluate the inner product by a tiny amount of measurements instead of numerous classical multiplications and summations.
Besides, as prior knowledge from experience and experts are informed, the parameter space is restricted to a finite-dimensional subspace while the expressiveness of the PQC is not reduced for the underlying problem. 
Hence it is rational to assume a mitigation of barren plateaus when compared to more general variational-based QML algorithms~\cite{Cerezo:2022challenges,martin:2024a,mcclean:2018barren}.

\textit{Conclusion} --
In this letter, we give an early yet natural glimpse of statistics models through the lens of quantum information and solve the aforementioned problems of preparing and learning statistics models on a quantum computer by introducing statistics information into the design of PQC.

As an efficient state-preparing subroutine of statistics models with varying parameters, SI-PQC can play an essential role in mitigating the intractable input bottleneck faced by practical QML and many other algorithms.
Moreover, any in-depth circuit simplification and compilation techniques can be shared, enabling the utilization of practical applications wherein speed speaks a lot such as quantitative and high-frequency trading.

As a trainable and knowledge-informed PQC,
SI-PQC can improve the trainability of QML algorithms such as GMM and physics-informed neural networks.
Furthermore, these parameters are learned simultaneously to reveal the intrinsic statistical properties in the given model, promising better theoretical interpretability.

In principle, SI-PQC can be extended to prepare time series, to study non-parametric and semi-parametric statistics tests, and to extract information from unknown quantum states.
We leave these interesting problems as future works.

\smallskip
This work is supported by National Key Research and Development Program of China (Grant No. 2023YFB4502500).

\newcommand{\beginsupplement}{%
	\setcounter{table}{0}%
	\renewcommand{\thetable}{S\arabic{table}}%
	\setcounter{figure}{0}%
	\renewcommand{\thefigure}{S\arabic{figure}}%
	\setcounter{equation}{0}%
	\renewcommand{\theequation}{S\arabic{equation}}%
	\setcounter{section}{0}%
	\renewcommand{\thesection}{\arabic{section}}%
        \setcounter{page}{1}%
        \renewcommand{\thepage}{S\arabic{page}}%
}

\let\oldaddcontentsline\addcontentsline
\renewcommand{\addcontentsline}[3]{}
\bibliographystyle{apsrev4-1}
\bibliography{main.bib}

\let\addcontentsline\oldaddcontentsline
\resetlinenumber
\clearpage
\onecolumngrid
\begin{center}
\textbf{\large Supplemental Material for \\
    ``Statistics-Informed Parameterized Quantum Circuit via Maximum Entropy Principle for Data Science and Finance''}
\end{center}

\maketitle

\beginsupplement
\renewcommand{\citenumfont}[1]{S#1}%
\renewcommand{\bibnumfmt}[1]{[S#1]}%

\NewDocumentCommand{\citesm}{>{\SplitList{,}} m }{%
  \def\temp{}%
  \ProcessList{#1}{\addSMprefix}%
  \expandafter\cite\expandafter{\temp}%
}
\newcommand{\addSMprefix}[1]{%
  \ifdefempty{\temp}%
    {\def\temp{SM_#1}}
    {\edef\temp{\temp,SM_#1}}%
}

\tableofcontents

\section{Related works}

\subsection{Works on statistics distribution  preparation task}

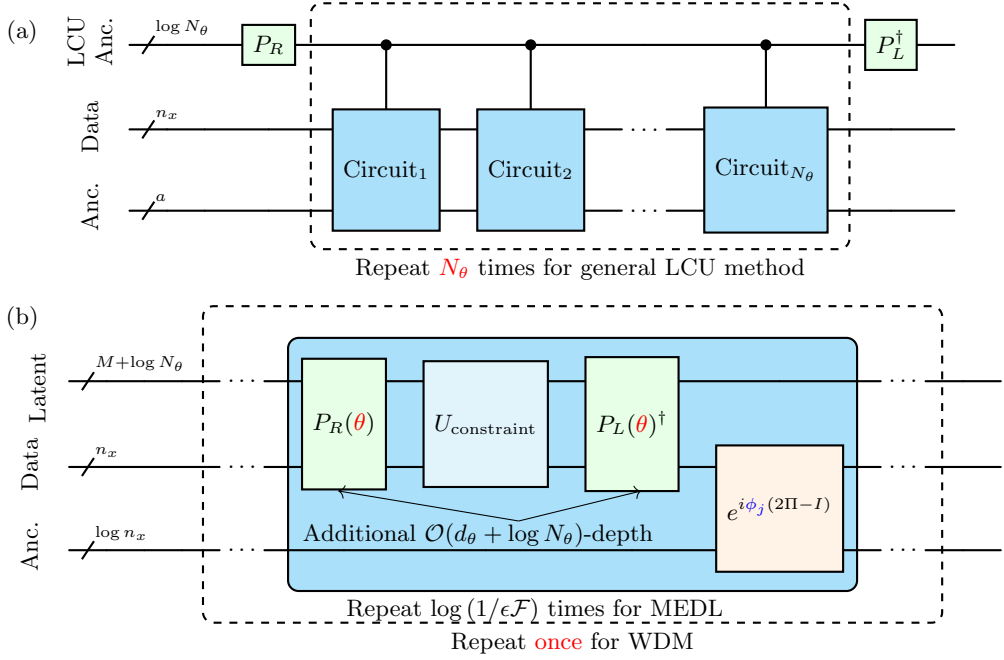
\begin{figure*}[htbp]
    \centering
    \begin{tikzpicture}
        \node at (0, -4.5) {        
        \begin{quantikz}
                \lstick[1, brackets=none,label style={yshift=0.5cm, xshift=-0.4cm, rotate=90}]{{Latent}}
                &\qwbundle{M+\log{N_\theta}}&&&\ \ldots\                 
                \gategroup[3,steps=6,background,style={dashed,rounded corners, inner xsep=3pt, inner ysep=16pt}, label style={label position=below,anchor=north,yshift=-0.2cm}]{Repeat \textcolor{red}{once} for WDM}
                &\gate[2, style={fill=green!10}]{P_R(\textcolor{red}{\theta})}
                \gategroup[3,steps=4,background,style={rounded corners, inner xsep=2pt, fill=cyan!30}, label style={label position=below,anchor=north,yshift=-0.17cm,xshift=-0.5cm}]{Repeat $\log{(1/\epsilon\mathcal{F})}$ times for MEDL}
                &\gate[2, style={fill=cyan!10}]{U_{\mathrm{constraint}}}&\gate[2, style={fill=green!10}]{P_L(\textcolor{red}{\theta})^\dagger}&&\ \ldots\ &&\\
                \lstick[1, brackets=none,label style={yshift=0.5cm, xshift=-0.5cm, rotate=90}]{{Data}}
                &\qwbundle{n_x}&&&\ \ldots\ &&&&\gate[2, style={fill=orange!10}]{e^{i\textcolor{blue}{\phi_j}(2\Pi-I)}}&\ \ldots\ &&\\
                \lstick[1, brackets=none,label style={yshift=0.5cm, xshift=-0.5cm, rotate=90}]{{Anc.}}
                &\qwbundle{\log{n_x}}&&&\ \ldots\ &                
                \lstick[1, brackets=none,label style={yshift=0.2cm, xshift=4.25cm}]{Additional $\mathcal{O}(d_\theta+\log{N_\theta})$-depth}
                &&&&\ \ldots\ &&
        \end{quantikz}
            };
        \draw[->] (0.1,-5) -- (1.73,-4.6);
        \draw[->] (0.1,-5) -- (-2.3,-4.6);
        \node at (0, 0) {
        \begin{quantikz}
            \lstick[1, brackets=none,label style={yshift=0.5cm, xshift=-0.5cm, rotate=90}]{{LCU\\Anc.}}
            &\qwbundle{\log{N_\theta}}&&\gate[1, style={fill=green!10}]{P_R}
            &\ctrl{1}
            \gategroup[3,steps=4,background,style={dashed,rounded corners, inner xsep=5pt, inner ysep=3pt}, label style={label position=below,anchor=north,yshift=-0.2cm}]{Repeat $\textcolor{red}{N_\theta}$ times for general LCU method}
            &\ctrl{1}&&\ctrl{1}&\gate[1, style={fill=green!10}]{P_L^\dagger}&\\
            \lstick[1, brackets=none,label style={yshift=0.5cm, xshift=-0.5cm, rotate=90}]{{Data}}
            &\qwbundle{n_x}&&&\gate[2, style={fill=cyan!30}]{\text{Circuit}_1}&\gate[2, style={fill=cyan!30}]{\text{Circuit}_2}&\ \ldots\ &\gate[2, style={fill=cyan!30}]{\text{Circuit}_{N_\theta}}&&\\
            \lstick[1, brackets=none,label style={yshift=0.5cm, xshift=-0.5cm, rotate=90}]{{Anc.}}
            &\qwbundle{a}&&&&&\ \ldots\ &&&
        \end{quantikz}
        };
        \node at (-6.5, 1.5) {\text{(a)}};
        \node at (-6.5, -2.3) {\text{(b)}};
    \end{tikzpicture}
    \caption{Comparison with related works for preparing a mixture of distributions.
    (a) For general preparing method, an $N_\theta$ times repetition of circuits is needed for $N_\theta$-components LCU.
    (b) For WDM method in this letter, the circuit is implemented once with totally $\mathcal{O}(d_\theta+\log{N_\theta})\log{(1/\epsilon\mathcal{F})}$ additional depth. 
    }
    \label{fig:compare}
\end{figure*}

The efficient amplitude encoding for continuous function is of general interest,
and many algorithms have been proposed to tackle this problem~\citesm{Grover:2002creating, holmes:2020efficient, Iaconis:2024quantum, zylberman:2023efficient, rattew:2022preparing, conde:2023efficient,rattew:2023non}.
Nevertheless, those algorithms fail to meet our requirements due to the common drawback of an implicit and indirect way to encode the statistics information:
The \emph{matrix product state (MPS)} method asks for an iterative procedure of piece-wise polynomial approximation, singular value decomposition, truncation, and gate transformation~\citesm{holmes:2020efficient, Iaconis:2024quantum}.
An alternative method of \emph{rank-1 projectors simulation} implements a time-dependent Hamiltonian adiabatic evolution wherein the statistics parameters are implicit~\citesm{rattew:2022preparing}.
Recently, the \emph{Walsh series loader (WSL)} method is proposed to prepare real-valued function demanding a classical computing procedure of Walsh series expansion coefficients, wherein the statistics information is intractable to extract as being diffused globally~\citesm{zylberman:2023efficient, conde:2023efficient}.
Another \emph{polynomial transformation of amplitudes} method combines the \emph{MPS} and \emph{WSL} methods to prepare the linear function, together with the quantum singular value transformation of amplitudes block-encoding to implement a (truncated) polynomial transformation of amplitudes~\citesm{guo:2021nonlinear, conde:2023efficient, rattew:2023non}.
While we are inspired by these elegant works and  our first algorithm utilizes the nonlinear amplitudes transformation technique in Refs.~\citesm{guo:2021nonlinear, rattew:2023non},
there are two essential differences between their works.

Firstly, in the quantum singular value transformation subroutine, there is a phase evaluating procedure.
This procedure is quite nontrivial as the original construction method admits numerical instability and an optimization procedure is suggested~\citesm{gilyen:2018quantum, yulong:2021efficient},
and this procedure has to be repeated many times when the model parameters are varying.
By introducing the maximum entropy principle we can prepare the distribution in a static-structured circuit when the model parameters are varying, and one does not need to re-evaluate the phase factor for the quantum singular value transformation subroutine.
Besides, we give an end-to-end analysis, in a unified way, of the preparation procedure for many often-used distributions, and extend the conclusions in Ref.~\citesm{rattew:2023non} to handle the case of more functions on larger intervals.

Secondly, since the original block-encoding framework is for operators while one needs to build a bridge from operators to amplitudes,
therein they develop an elegant way named \emph{block-encoding of amplitudes} to realize this~\citesm{guo:2021nonlinear, conde:2023efficient, rattew:2023non}:
\begin{theorem*}[\textbf{Block-encoding of amplitudes, specified from Theorem~4 of Ref.~\citesm{guo:2021nonlinear}}]\nonumber\label{lem:bea}
    
    Suppose a state preparation unitary $U$, one can construct a unitary $\vec{G}$ such that
    \begin{equation}
        (\bra{0}\otimes\mathcal{I}_{2n+1})\vec{G}(\ket{0}\otimes\mathcal{I}_{2n+1}) = \sum_{j=0}^{N-1}x_j\ket{\phi_j}\bra{\phi_j},
    \end{equation}
    by using controlled-$U$ and controlled-$U^\dagger$ four times and another $\mathcal{O}(n)$ one- and two-qubit gates.
\end{theorem*}
\noindent A constructive proof of this result is given in Ref.~\citesm{guo:2021nonlinear}, and it also originates from the quantum analog-digital conversion algorithm~\citesm{mitarai:2019quantum}.
Nevertheless, the number of ancillary qubits scales as $\mathcal{O}(n)$ and the corresponding success probability of measuring the target state on the initial state $H^{\otimes n}\ket{0}$ scales as
$\sum_x \left(\frac{p(x)}{\sqrt{2^n}}\right)^2 = \frac{1}{2^n}$, inducing an exponentially growing time complexity.
In our work, we introduce a linear block-encoding
instead of an amplitude block-encoding with $\mathcal{O}(\log{n})$ ancillary qubits and $\mathcal{O}(1)$ probability as proved in the following text.
Consequently, we argue that our work can make an exponential reduction on both the number of ancillary qubits and the time complexity when focusing on continuous function preparation problem in comparison to Ref.~\citesm{guo:2021nonlinear, rattew:2023non}.

To compare with problem-agnostic PQC-based quantum algorithms, there are three main differences:

Firstly, we consider the scalability and trainability problems to prepare a distribution on $n$ qubits.
A general problem agnostic PQC suffers from scalability and trainability problems when system size grows.
For example, in Ref.~\citesm{zoufal:2019quantum} a generative model is implemented via a variational quantum circuit on $n$ qubits with $\mathcal{O}(n)$ parameterized gates to learn the distribution.
In our SI-PQC scheme, the independent circuit parameters are fixed to be $M$, the free parameters number of the underlying model, and are invariant when enlarging the system size $n$.
Therefore, we argue that SI-PQC can offer a potential improvement in the scalability and trainability with growing system size $n$, corresponding to the sampling precision.

Secondly, we consider the time and resource consumption problem in the training procedure.
When preparing an explicitly encoded distribution with an analytic formula, SI-PQC needs no quantum training procedure and the classical circuit compiling procedure for the constraint and exponentiation layers can both be re-used since they are just elementary functions.
For varying model parameters, one only needs to re-compute $M$ rotation angles by the explicit formula given above.
However, other problem-agnostic PQC-based algorithms need to be re-trained for an identical distribution with varying parameters.
Consequently, we argue that SI-PQC shows a feasible way to reduce the time and resource consumption, as well as better transferability on statistics states preparation tasks.

\subsection{Works on statistics mixture  preparation task}
In principle, the mixture of distribution can be prepared by combining of LCU technique and all of those distribution-preparing methods mentioned above.
Nevertheless, these methods turn out to be inefficient when preparing a mixture of numerous distributions:
Despite pre-computation procedure, the circuit has to be repeated $N_\theta$ times when preparing an $N_\theta$-mixture by LCU as depicted in Fig.~\ref{fig:compare}(a).

For the WDM method in this letter, since the circuit is static and can be reused, the circuit can be implemented only once with only $\mathcal{O}(d_\theta+\log{N_\theta})\log{(1/\epsilon\mathcal{F})}$  additional circuit depth as depicted in Fig.~\ref{fig:compare}(b).
Consequently, we argue that WDM enables an exponential time complexity reduction in terms of latent parameter space size (corresponding to the distribution numbers) over the aforementioned works on the task of preparing a distribution mixture.

\subsection{Works on statistics model learning task}
The statistics model learning is also of general interest, especially in the fields of machine learning and statistics.
The learning of distribution mixture based on quantum linear algebra is studied in Ref.~\citesm{hideyuki:2020quantum}.
Therein, they propose a powerful algorithm named quantum expectation-maximization to solve the GMM learning problem assuming quantum access to the GMM defined to be
\begin{definition*}[\textbf{Quantum access to GMM in Ref.~\citesm{hideyuki:2020quantum}}]\nonumber
   We say that we have quantum access to a GMM  of a dataset $V\in \mathbb{R}^{n\times d}$ and model parameters $\theta_j\in\mathbb{R},\mu_j\in\mathbb{R}^d,\Sigma\in\mathbb{R}^{d\times d}$ for all $j\in[k]$ if we can perform in $\mathcal{O}(polylog (d))$ the following maps:
   \begin{itemize}
       \item $\ket{j}\ket{0}\mapsto\ket{j}\ket{\mu_j}$,
       \item $\ket{j}\ket{i}\ket{0}\mapsto\ket{j}\ket{i}\ket{\sigma_i^j}$ for $i\in[d]$, where $\sigma_i^j$ is the $i$-the rows of $\Sigma_j\in\mathbb{R}^{d\times d}$,
       \item $\ket{i}\ket{0}\mapsto\ket{i}\ket{v_i}$ for all $i\in[n]$,
       \item $\ket{i}\ket{0}\ket{0}\mapsto\ket{i}\ket{vec[v_iv_i^T]}=\ket{i}\ket{v_i}\ket{v_i}$ for $i\in[n]$,
       \item $\ket{j}\ket{0}\mapsto\ket{j}\ket{\theta_j}$.
   \end{itemize}
   For instance, one may use a QRAM data structure.
\end{definition*}
In our work, we can learn the GMM directly without the assumption of the QRAM. 
Instead, SI-PQC can serve as the quantum access to GMM for the quantum expectation-maximization algorithm in Ref.~\citesm{hideyuki:2020quantum} after some modification (for example, a standard quantum amplitude estimation subroutine).
Furthermore, SI-PQC can also be utilized to implement more general statistics learning tasks, such as parametric tests, non-parametric tests, and semi-parametric tests.

On the other hand, those state preparation algorithms discussed above can not be utilized to solve learning problems:
While variational quantum circuits can learn the implicit distribution of data as discussed in Ref.~\citesm{zoufal:2019quantum},
the statistics model parameters can not be extracted directly since the gate parameters can not be converted into the model parameters. And there is no guarantee that the output state should be a GMM or other specific models as wanted.
In this sense, we consider SI-PQC as a problem-inspired candidate with better interpretability.


\section{Preparing Statistics Distribution via Maximum Entropy Distribution Loader}\label{sec:med-SI-PQC}

In this section, we address the aforementioned issue of preparing the quantum state of given statistic distributions formally defined as
\begin{definition}[\textbf{Quantum state preparation of a statistics distribution}]\label{def:distribution}
    Suppose a statistics distribution or its discretization with the probability density function (p.d.f.) to be $p(x)$ and $\sum_x p(x) = 1$.
    Then the quantum state preparation of this distribution is defined to construct the oracle
    \begin{equation}
        \mathcal{O}_p\ket{0} = \sum_x \vec{p}(x)\ket{x},
    \end{equation}
    wherein $\vec{p}(x)$ is the probability amplitude corresponding to the normalized p.d.f. as
    \begin{equation}
        \vec{p}(x) = \frac{p(x)}{\sqrt{\sum_k p(k)^2}}.
    \end{equation}
\end{definition}

\subsection{Maximum Entropy Principle}\label{sec:maximum_entropy}

Intuitively, the maximum entropy principle states that given prior knowledge of the system, such as the expectations of some observables, 
the best guess for the underlying distribution is the one with maximum entropy~\citesm{conrad:2004probability, park:2009maximum, banavar:2010applications}.
Formally, we have:
\begin{lemma}[\textbf{Maximum entropy principle for statistics distributions with observable expectation constraints}]\label{lem:med}
    
    Suppose $M$ constraints of the expectation of observables $\mathbb{E}[f_k(x)] = a_k (1\le k \le M)$, then the maximum entropy distribution admits a probability density function (p.d.f.)
    \begin{equation}\label{eq:pdf_me}
        p(x) = \frac{1}{\mathcal{N}}\exp{\sum_{k=1}^M \lambda_k f_k(x)},
    \end{equation}
    wherein $f_0 = 1$ is the constant function, $\mathcal{N}$ is the normalization constant and $\lambda_k$ are the Lagrange multipliers satisfying
    \begin{equation}
        \max_{\lambda}\{\sum_{k=1}^M\lambda_k a_k - \int \exp{\sum_{k=1}^M \lambda_k f_k(x)}\mathrm{d}x \}.
    \end{equation}
\end{lemma}
\begin{proof}
    One can consider the functional derived from the Lagrange multiplier method as
    \begin{equation}
        \begin{split}
            F(p) = &\int p(x)\ln{p(x)}dx\\
            &- \lambda_0(\int p(x)dx-1)\\
            &- \sum_{k=1}^M\lambda_k(\int f_k(x)dx - a_k).
        \end{split}        
    \end{equation}
    Since $p(x)$ maximizes the entropy, the functional derivatives satisfy
    \begin{equation}
        \frac{\delta F}{\delta p} = \ln{p(x)} + 1 -\lambda_0 - \sum_{k=1}^M\lambda_kf_k(x).
    \end{equation}
    For convenience, one can denote $f_0=1$ and substitute $\lambda_0\leftarrow\lambda_0-1$ to finish the proof.
\end{proof}

\noindent From the quantum computation perspective, these coefficients can either be directly computed by the model parameters for common distributions (see also Ref.~\citesm{park:2009maximum}), or be derived from an optimization procedure as discussed later.

The generality of the maximum entropy principle is twofold:
Firstly, one noticed that the distribution of every statistic can be treated as a maximum entropy one under some constraints.
Secondly and even more importantly, the constraints are usually the expectation of some elementary function for most practical statistics distributions,
including Gaussian, Laplace, Pareto, Beta, Cauchy, Chi-squared, and log-normal distributions, to name but a few~\citesm{conrad:2004probability, park:2009maximum}.
These two propositions are formally stated as:
\begin{lemma}[\textbf{Generality of maximum entropy distribution}]\label{lem:gmed}
    
    \textbf{(a)} Any distribution is a maximum entropy distribution.
    \textbf{(b)} Most often-used distribution can be viewed as a maximum entropy distribution with a finite number of constraints taking the form of elementary function expectation (summarized in Table.~\ref{tab:medl}).
\end{lemma}
\begin{proof}
    \textbf{(a)} 
    Suppose a distribution with p.d.f. $p(x)$.
    One can construct a constraint that the expectation of the observable $\ln{p(x)}$ equals the constant $a = \int \ln{p(x)}p(x) dx$.
    \textbf{(b)} 
    A direct computation works, following the combination of Lemma~\ref{lem:gmed} and the constraints listed in Table.~\ref{tab:medl}.
\end{proof}

    \renewcommand\arraystretch{2}
\begin{table*}[htbp]    
\begin{threeparttable}
    \centering
    \caption{Summary of maximum entropy distribution, constraints, and preparation time complexity.}
    \label{tab:medl}
    \begin{tabular}{|l|c|c|c|}
    \hline
        Distribution family & Probability density function & Constraints & Total complexity \\
    \hline
        Normal &
        $\frac{1}{\sqrt{2\pi}\sigma}\exp\left( {-\frac{(x-\mu)^2}{2\sigma^2}} \right)$ & 
        $\mathbb{E}[x]=\mu$, $\mathbb{E}[(x-\mu)^2]=\sigma^2$ &
        $\mathcal{O}(\frac{n}{\sigma\sqrt{\sigma^2+\mu^2}}\log{\frac{1}{\sigma\sqrt{\sigma^2+\mu^2}\epsilon}})$\\
    \hline
        Exponential & 
        $\lambda \exp(-\lambda x)$, $\lambda>0, x\geq0$
        & $\mathbb{E}[x]=\frac{1}{\lambda}$
        & $\mathcal{O}(\frac{n}{\lambda^2}\log{(\lambda/\epsilon)})$\\
    \hline
        Pareto&
        $\frac{\alpha x_0^\alpha}{x^{\alpha+1}}$&
        $\mathbb{E}[\ln{x}]=\frac{1}{\alpha}+\ln{x_0}$&
        $\mathcal{O}\left(n\log^2(\frac{1}{\epsilon})\right)$\\
    \hline
        Rayleigh&
        $\frac{x}{\sigma^2}\exp\left( {-\frac{x^2}{2\sigma^2}} \right)$ on $[0, x_0]$&
        $\mathbb{E}[x^2]=2\sigma^2$, $\mathbb{E}[\ln{x}]=\frac{\ln(2\sigma^2)-\gamma_E}{2}$&
        $\mathcal{O}\left(\frac{n}{\mathcal{F}}\text{polylog}(\frac{1}{\mathcal{F}}, \frac{1}{\epsilon})\right), \mathcal{F} = \sqrt{\frac{e}{2x_0}\gamma(2, \frac{x_0^2}{\sigma^2})}$        
        \\
    \hline
        Chi&
        $\frac{2}{2^{k/2}\Gamma(k/2)}x^{k-1}\exp\left( {-\frac{x^2}{2}} \right)$&
        $\mathbb{E}[x^2]=k$, $\mathbb{E}[\ln{x}]=[\psi(\frac{k}{2})+\ln{2}]/2$&
        $\mathcal{O}\left(\frac{n}{\mathcal{F}}\text{polylog}(\frac{1}{\mathcal{F}}, \frac{1}{\epsilon})\right), \mathcal{F} 
            =\frac{\sqrt{k}\Gamma(\frac{k}{2})}{\sqrt{2}(\frac{k-1}{2e})^{\frac{k-1}{2}}}$
        \\
    \hline
        Chi-squared&
        $\frac{2}{2^{k/2}\Gamma(k/2)}x^{k/2-1}\exp\left( {-\frac{x}{2}} \right)$&
        $\mathbb{E}[x]=k$, $\mathbb{E}[\ln{x}]=\psi(\frac{k}{2})+\ln{2}$&
        $\mathcal{O}\left(\frac{n}{\mathcal{F}}\text{polylog}(\frac{1}{\mathcal{F}}, \frac{1}{\epsilon})\right), \mathcal{F} = \frac{2\sqrt{k(k+2)}\Gamma(\frac{k}{2})}{\sqrt{2}(\frac{k-1}{2e})^{k/2-1}}$
        \\
    \hline
    \end{tabular}
\end{threeparttable}
\end{table*}
\noindent Notice that the maximum entropy principle reveals an elegant and unified way to encode the statistics information.
There are no restrictions on the constraints' form $f_k$, while it is already quite general by considering the family of arithmetic and elementary functions that can be efficiently computed.
Consequently, we can develop a natural and efficient paradigm to inject statistics information into a PQC in three steps (as shown in Fig.~\ref{fig:medl}):

\subsection{Prepare distinct constraint functions}\label{sec:constraint}
Firstly, for each constraints $\mathbb{E}[f_k] = a_k$, one can prepare the corresponding oracle $O_k$ of function $f_k$:
\begin{equation}
    O_k: \ket{0} \rightarrow \frac{1}{{\lVert f_k\rVert}_2}\sum_{j=0}^{N-1}f_k(x_j)\ket{j},
\end{equation}
wherein ${\lVert f_k\rVert}_2 = \sqrt{\sum_{j=0}^{N-1}f_k(x_j)^2}$ is the ${l}_2$-norm of $f_k$.
To derive other distributions, one can implement the nonlinear transformation with the block-encoding language to characterize this procedure at the cost of $a$ ancillary qubits and an error bound $\epsilon$:
\begin{definition}[\textbf{Block-encoding, Definition 43 in Ref.~\citesm{gilyen:2018quantum}}]
     Suppose an $n$-qubit operator $A$, a scaling constant $\alpha\in\mathbb{R}_+$, an ancillary qubit number $a\in\mathbb{N}$, and an error bound $\epsilon$. The $(a+n)$-qubit unitary is defined to be an $(\alpha, a, \epsilon)$-block-encoding of $A$ if
    \begin{equation}
        \lVert A-\alpha(\bra{0}^{\otimes a}\otimes\mathcal{I}_n)U(\ket{0}^{\otimes a}\otimes\mathcal{I}_n)\rVert \leq \epsilon.
    \end{equation}
\end{definition}
\noindent One can apply the linear combination of unitaries (LCU) technique that was first introduced to simulate Hamiltonian in Ref.~\citesm{childs:2012hamiltonian} and then reformulated in Ref.~\citesm{gilyen:2018quantum},
and the quantum singular value transformation technique in Ref.~\citesm{gilyen:2018quantum} to implement the desired nonlinear transformation.
More specifically, we will use the following results (as depicted in Fig.~\ref{fig:real_polynomial_eigenvalue_transformation}):
\begin{lemma}[\textbf{Real polynomial eigenvalue transformation of arbitrary parity, reformulated from Theorem~56 in Ref.~\citesm{gilyen:2018quantum}}]\label{lem:rpet}
    
    Suppose an $(\alpha, a, \epsilon)$-block-encoding of a Hermitian matrix $A$, 
    a $d$-degree real polynomial $f: \mathbb{R}\rightarrow\mathbb{R}$ bounded on the interval $[-1, 1]$ with $\lVert P\rVert_{\mathrm{max}}\leq\frac{1}{2}$.
    Then there is a $(1, a+2, 4d\sqrt{\epsilon/\alpha}+\delta)$-block-encoding of $f(A/\alpha)$ with $d$ queries of $U$ and $U^\dagger$, one query of controlled-$U$ and $\mathcal{O}(ad)$ other one- and two-qubit gates.
\end{lemma}
\begin{lemma}[\textbf{Linear combination of block-encoded-matrices, Reformulated from Definition 51 and Lemma~52 in Ref.~\citesm{gilyen:2018quantum}}]\label{lem:lcbem}
    
    Let: (1) $y\in \mathbb{C}^m$, $\lVert y\rVert_1\le \beta$, and $(P_L, P_R)$ a $(\beta, b, \epsilon_1)$ state-preparation-pair of unitaries such that: $P_L\ket{0}_b=\sum_{j=0}^{2^b-1}c_j\ket{j}$, $P_R\ket{0}_b=\sum_{j=0}^{2^b-1}d_j\ket{j}$,
        $\sum_{j=0}^{m-1}|\beta(c_j^*d_j)-y|\le\epsilon_1$,
        and $c_j^*d_j=0$ for all $j\ge m$;
    (2)$A = \sum_{j=0}^{m-1}y_j A_j$ be an $s$-qubit operator, $W=\sum_{j=0}^{m-1}\ket{j}\bra{j}\otimes U_j + (I - \sum_{j=0}^{m-1}\ket{j}\bra{j}\otimes U_j)\otimes I_a\otimes I_s$ is an $s + a + b$ qubit unitary such that for all $j<m$ we have that $U_j$ is an ($(\alpha, a, \epsilon_2$)-block-encoding of $A_j$.
    Then we implement an $(\alpha\beta, a+b, \alpha\epsilon_1 + \alpha\beta\epsilon_2)$-block-encoding of $A$, with a single use of $W$, $P_R$, and $P_L^\dagger$.    
\end{lemma}
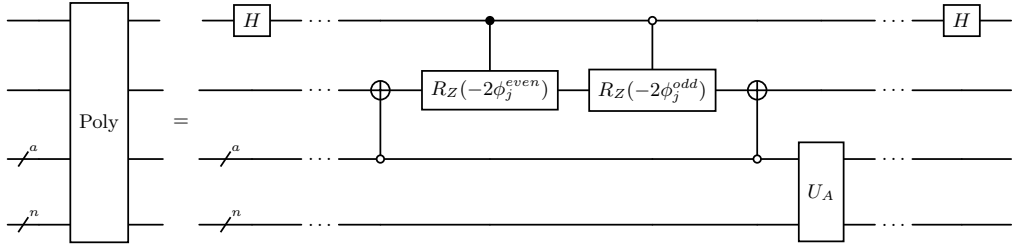
\begin{figure*}
\begin{adjustbox}{width=0.75\linewidth}
    \begin{quantikz}
        &
        &\gate[4]{\mathrm{Poly}}&\midstick[4,brackets=none]{=}&\gate{H}&\ \ldots\ &&\ctrl{1}&\ctrl[open]{1}&&&\ \ldots\ &\gate{H}&\\
        &&&&&\ \ldots\ &\targ{}&\gate{R_Z(-2\phi_j^{even})}&\gate{R_Z(-2\phi_j^{odd})}&\targ{}&&\ \ldots\ &&\\
        &\qwbundle{a}&&&\qwbundle{a}&\ \ldots\ &\ctrl[open]{-1}&&&\ctrl[open]{-1}&\gate[2]{U_A}&\ \ldots\ &&\\
        &\qwbundle{n}&&&\qwbundle{n}&\ \ldots\ &&&&&&\ \ldots\ &&\\
    \end{quantikz}
\end{adjustbox}
    \caption{Quantum circuit to implement real polynomial eigenvalue transformation via QSVT.
    }
    \label{fig:real_polynomial_eigenvalue_transformation}
\end{figure*}
Noticed that Lemma~\ref{lem:rpet} acts on the eigenvalues while a nonlinear transformation on the amplitudes is required, we need to construct a block-encoding of the linear function as
\begin{lemma}[\textbf{Block-encoding of linear function}]\label{lem:lft}
     There is a $(1, \log{n}, 0)$ block-encoding of the $2^n$-dimension diagonal matrix $A = diag(0, \frac{1}{2^n-1}, \frac{2}{2^n-1}, \frac{3}{2^n-1}, ..., 1)$ with $\mathcal{O}(n)$ circuit depth.
\end{lemma}
\begin{proof}
    One observes that
    \begin{equation}
        A = I - \sum_{j=1}^n \frac{2^{j-1}}{2^n}I^{\otimes(j-1)}\otimes Z\otimes I^{\otimes(n-j)}.
    \end{equation}
    Hence one implement a linear combination of $n$ numbers of $-Z$ gates controlled by the $\log{n}$ ancillary qubits for the state-preparation pair $P_L, P_R$ defined as
    \begin{equation}
            P_L\ket{0}_{\log{n}} = P_R\ket{0}_{\log{n}} = \sum_{j=0}^n\sqrt{\frac{2^{j-1}}{2^n}}\ket{j}.
    \end{equation}
    Given the sampling space size $2^n$, $P_L$ and $P_R$ can be implemented by a standard amplitude-encoding subroutine via uniformly controlled Y-rotation with circuit depth $\mathcal{O}(n)$.
    To summarize, the total circuit depth is still $\mathcal{O}(n)$.
\end{proof}
Combining QSVT and block-encoding of the linear function, we can improve and generalize the results in Refs.~\citesm{conde:2023efficient, guo:2021nonlinear, rattew:2023non} as:
\begin{lemma}[\textbf{Efficient block-encoding of non-linear functions, improved from Theorem~1 in Ref.~\citesm{conde:2023efficient}, Theorem~5 in Ref.~\citesm{guo:2021nonlinear}, and Theorem~5 in Ref.~\citesm{rattew:2023non}}]\label{lem:nta}

    Suppose that one is given access to a state preparation oracle $U_\psi$ of the $n$-qubits state $U_\psi\ket{0}=\ket{\psi}=\sum_{j=0}^{N-1}\psi_j\ket{j}$ and a function $f$, then the non-linear transformation of amplitudes $f(U_\psi)=\frac{1}{\mathcal{N}}\sum_{j=0}^{N-1}f(\psi_j)\ket{j}$ can be $\epsilon$-approximated with $\log{n}+2$ ancillary qubits, and controlled $U_\psi$ and $U_\psi^\dagger$ query complexity to be:
    \begin{itemize}
        \item $\mathcal{O}(d/\mathcal{F})$ if $f$ is a $d$-degree polynomial with filling rate $\mathcal{F} = \frac{\lVert f\rVert_2}{\lVert f\rVert_{\mathrm{max}}}$;
        \item $\mathcal{O}(\log{\frac{1}{\epsilon}})$ if $f=e^x$ or $\cos{x}$;
        \item $\mathcal{O}(n\log{\frac{1}{\epsilon}})$ if $f=\sin{x}$;
        \item $\mathcal{O}(\sigma\log{\frac{1}{\epsilon}})$ if $f=\frac{1}{\sigma\sqrt{2\pi}}e^{-\frac{1}{2\sigma^2}x^2}$ with $\sigma^2\ge1/2$;
        \item $\mathcal{O}({\log(1/\epsilon)\frac{\ln(1-c)}{\ln(1-c/2)}})$ if $f=\ln{(1+cx)}$ with $|c|<1$.
    \end{itemize}
    Furthermore, in the function preparation task, the unitary of interest is $U_L$ as defined in Lemma~\ref{lem:lft} with gate complexity $\mathcal{O}(n)$, and hence the total gate complexity needs another multiplier of $n$.
\end{lemma}
\begin{proof}    
    The proof consists of two steps:
    Firstly, one considers the case of a $d$-degree real-value polynomial $P(x): \mathbb{R}\rightarrow\mathbb{R}$, 
    and re-scale it to $\vec{P}(x) = P(x)/2{\lVert P(x)\rVert}_{\mathrm{max}}$ to satisfy the condition of Lemma~\ref{lem:rpet}.
    In our situation, the error bound is set to be $\delta=\frac{\epsilon}{2^{n+1}{\lVert P(x)\rVert}_{\mathrm{max}}}$ and one derives the quantum singular value transformation unitary $U_P$ so that
    \begin{equation}
        \begin{split}
            &\left\lVert P(x_k) - 2{\lVert P(x)\rVert}_{\mathrm{max}}U_P(x_k) \right\rVert\\
            =&2{\lVert P(x)\rVert}_{\mathrm{max}}\lVert \vec{P}(x_k)-U_p(x_k) \rVert\\
            \leq&2{\lVert P(x)\rVert}_{\mathrm{max}}\delta=\frac{\epsilon}{2^n{\lVert P(x)\rVert}_{\mathrm{max}}}.
        \end{split}
    \end{equation}
    One then applies $U_p$ to the linear block-encoding $A$ as defined in Lemma~\ref{lem:lft}.
    Due to Lemma~\ref{lem:rpet}, $\vec{G}$ is queried $\mathcal{O}(d)$ times,
    and each query possesses a gate complexity $\mathcal{O}(\frac{{\lVert P(x)\rVert}_{\mathrm{max}}}{{\lVert P(x)\rVert}_{2}})$.
    The total complexity is $\mathcal{O}(d/\mathcal{F})$ as claimed.
    
    Secondly, one considers the general case of the nonlinear function. 
    Notice that any continuous function $f(x)$ defined on a closed interval can be arbitrarily approximated by a polynomial $P_f(x)$ such that $\lVert f(x)-P_f(x)\rVert\leq \epsilon$.
    One can choose the polynomial $P_f$ to be a $\frac{\epsilon}{2}$-approximation of $f$, and set $\epsilon=\epsilon/2$ in the first step.
    More specific results can be derived as in Ref.~\citesm{rattew:2023non}:
    
    \noindent\textbf{(a)} For $f=e^x$, consider the $k$-degree polynomial $P_f^{(k)} = \sum_{j=0}^k\frac{1}{j!}x^j$ with error bound $1/2^k$.
    Since $e^x$ is uniformly bounded by $1/e$ and $e$ on $[-1, 1]$, the query complexity $\mathcal{O}(d/\mathcal{F})$ is $\mathcal{O}(\log{(1/\epsilon)}/e^2)$, i.e., $\mathcal{O}(\log{(1/\epsilon)})$.
    
    \noindent\textbf{(b)} For $f=\cos{x}$, consider the $k$-degree polynomial $P_f^{(k)} = \sum_{j=0}^k (-1)^j\frac{1}{(2j)!}x^{2j}$ with error bound $1/2^k$.
    Since $\cos{x}$ is uniformly bounded by $\cos{1}$ and $1$ on $[-1, 1]$, the query complexity $\mathcal{O}(d/\mathcal{F})$ is $\mathcal{O}(\log{(1/\epsilon)}/\cos{1})$, i.e., $\mathcal{O}(\log{(1/\epsilon)})$.
    
    \noindent\textbf{(c)} For $f=\sin{x}$, consider the $k$-degree polynomial $P_f^{(k)} = \sum_{j=0}^k (-1)^j\frac{1}{(2j+1)!}x^{2j+1}$ with error bound $1/2^k$.
    Since $\lVert\sin{x}\rVert_{\mathrm{max}} = \sin{1}$ and $\lVert\cos{x}\rVert_{2}$ is lower bounded by $\sqrt{\frac{\int_{-1}^1 \sin{x}^2 dx}{2}}$ on $[-1, 1]$, the query complexity $\mathcal{O}(d/\mathcal{F})$ is still $\mathcal{O}(\log{(1/\epsilon)})$.
    
    \noindent\textbf{(d)} For $f=\frac{1}{\sigma\sqrt{2\pi}}e^{-\frac{1}{2\sigma^2}x^2}$, consider the $k$-degree polynomial $P_f^{(k)} = \frac{1}{\sigma\sqrt{2\pi}}\sum_{j=0}^k\frac{(-\frac{1}{2\sigma^2})^j}{j!}x^{2j}$ with error bound $1/\pi2^k$.
    Since $\lVert f(x)\rVert_{\mathrm{max}} = 1$ and $\lVert f(x)\rVert_{2}$ is lower bounded by $\frac{1}{\sigma\sqrt{2\pi}}$ on $[-1, 1]$, the query complexity $\mathcal{O}(d/\mathcal{F})$ is $\mathcal{O}(\sigma\log{(1/\epsilon)})$.

    \noindent\textbf{(e)} For $f=\ln(1+cx)$ with $0<c<1$, consider the $k$-degree polynomial $P_f^{(k)} = \sum_{j=1}^k\frac{(-1)^{j+1}(cx)^j}{j}$ with error bounded by $1/c^k$.
    Since $\lVert f(x)\rVert_{\mathrm{max}} = \ln{\frac{1}{1-c}}$ and
    \begin{align*}
        \lVert f(x)\rVert_{2}=&\sqrt{\frac{\int_{-1}^1 \ln^2{(1+cx)} dx}{2}}
        =\sqrt{\frac{1}{2c}\int_{1-c}^{1+c}\ln^2udu}\notag\\
        \geq&\sqrt{\frac{1}{2c}\int_{1-c}^{1-c/2}\ln^2udu}
        \geq-\frac{1}{2}\ln(1-c/2),
    \end{align*}
    the query complexity $\mathcal{O}(d/\mathcal{F})$ is bounded by $\mathcal{O}({\log(1/\epsilon)\frac{\ln(1-c)}{\ln(1-c/2)}})$.
\end{proof}
\noindent\textbf{Remark.} More detailed, to make the function implementable by QSVT, one indeed prepare
\begin{equation}
    \vec{f_k} = \frac{f_k}{2{\lVert f_k\rVert}_{\mathrm{max}}}
\end{equation}
with query complexity $\mathcal{O}(d_k\frac{{\lVert f_k\rVert}_{\mathrm{max}}}{{\lVert f_k\rVert}_2})$.
It should be emphasized that for standard elementary functions such as constant, power, trigonometric, exponential, and logarithms functions,
the corresponding subroutine can be viewed as a static-structured subcircuit (depicted in the blue boxes in Fig.~\ref{fig:medl} (b)).
Furthermore, any in-depth circuit simplification or compilation can be introduced and cured to improve the practical execution performance.

At the end of this subsection, we give a tiny result that will be used in the following sections to enable an end-to-end error analysis.
\begin{lemma}[\textbf{QSVT error bound with imperfect rotation angles}]\label{lem:qsvt_rotation_error}
     Suppose a QSVT with $d$ rotation angle phases, wherein each phase is imperfect with an error bounded by $\Delta\theta$. Then the total error is bounded by $d\Delta\theta$.
\end{lemma}
\begin{proof}
    In brief, a QSVT circuit can be written in the form of
    \begin{equation}
        U(\theta_1, \theta_2, ..., \theta_d) = U_d Z(\theta_d) U_{d-1} Z(\theta_{d-1}) ... U_1 Z(\theta_1) U_0.
    \end{equation}
    Then the error of the imperfect implementation can be bounded as:
    \begin{align}
        &|U(\theta_1, \theta_2, ..., \theta_d) - U(\Tilde{\theta_1}, \Tilde{\theta_2}, ..., \Tilde{\theta_d})|\\
        \leq&|U(\theta_1, \theta_2, ..., \theta_d) - U(\Tilde{\theta_1}, {\theta_2}, ..., {\theta_d})|\nonumber\\
        &+|U(\theta_1, \theta_2, ..., \theta_d) - U({\theta_1}, \Tilde{\theta_2}, ..., {\theta_d})|\nonumber\\
        &+|U(\theta_1, \theta_2, ..., \theta_d) - U({\theta_1}, {\theta_2}, ..., \Tilde{\theta_d})|\\
        \leq&d\Delta\theta.
    \end{align}
\end{proof}

\subsection{Implement parameterized combination of constraints}\label{sec:lcc}
One can implement the parameterized combination of constraints via LCU.
According to Lemma~\ref{lem:lcbem}, one needs an efficient construction of state-preparation-pair for the constraint coefficients.
Furthermore, to exposure these coefficients, one has a trade-off to encode them sparsely with $M$ qubits,
instead of a dense encoding of $\lceil{\log_2{M}}\rceil$ qubits,
and derive the following explicit construction of state-preparation-pair unitaries.
\begin{lemma}[\textbf{Efficient construction of sparse state-preparation-pair}]\label{lem:spp}
     Let $y \in \mathbb{R}^M$ be the $M$ coefficients, then there is a (${\lVert y\rVert}_1, M-1, \epsilon_\mathrm{spp}$) state-preparation pair of unitaries $(P_L, P_R)$, wherein $\epsilon_\mathrm{spp}$ is the corresponding error from implement an arbitrary angle rotation.
\end{lemma}
\begin{proof}
    The proof for positive $y\in\mathbb{R_+}^M$ is direct: observe that the $M$-qubit operator $P(\vec{\theta})$ (illustrated in Fig.~\ref{fig:state_preparation_pair}) that consists of a ladder of rotation and controlled-rotation gates
    \begin{multline}
        P(\vec{\theta})=C_{M-1}X_{M-2}\circ C_{M-2}Y(\theta_{M-1})_{M-1}\circ \hdots\circ \\
        C_1X_0\circ C_{0}Y(\theta_{1})_{1}\circ Y(\theta_{0})_{0}
    \end{multline}
    can generate a sparse superposition of $M$ bases $\{|{\underbrace{00...0}_{M}}\rangle, |{1\underbrace{0...0}_{M-1}}\rangle, |{01\underbrace{0...0}_{M-2}}\rangle, ..., |{\underbrace{00...0}_{M-1}1}\rangle\}$ as
    \begin{equation}
        \begin{split}
            P(\vec{\theta})\ket{0}_{M}
            &=\cos{\frac{\theta_0}{2}}|{\underbrace{00...0}_{M}}\rangle
            +\sin{\frac{\theta_0}{2}}\cos{\frac{\theta_1}{2}}|{1\underbrace{0...0}_{M-1}}\rangle\\
            +&...+\sin{\frac{\theta_0}{2}}\sin{\frac{\theta_1}{2}}\hdots\sin{\frac{\theta_{M-1}}{2}}|{\underbrace{00...0}_{M-2}1}\rangle.
        \end{split}        
    \end{equation}
    Then the state-preparation-pair can be constructed by setting the rotation angles as:
    \begin{equation}
        \begin{split}
            &P_L = P_R = P(\theta)\\
            &\theta_j = 2\arctan{\sqrt{
            \frac{\sum_{k=j+1}^{M-1}|y_k|}{|y_j|}}} (0\leq j\leq M-2).
        \end{split}
    \end{equation}    
    The verification is direct:
    \begin{align}
        P_L\ket{0} =& P_R\ket{0} =\sum_{j=0}^{M-1} \frac{\sqrt{y_j}}{{\sqrt{\lVert y\rVert}_1}}|{\underbrace{0...0}_{M-1-j}}1{\underbrace{0...0}_{j-1}}\rangle,\\
        {\lVert y\rVert}_1&(\frac{\sqrt{y_j^*}}{{\sqrt{\lVert y\rVert}_1}}\frac{\sqrt{y_j}}{{\sqrt{\lVert y\rVert}_1}}) - y_j = 0.
    \end{align}
    The case with negative $y_j$ is more tricky:
    If $y_{j-1}$ and $y_j$ have opposite signs, then one can implement X rotations with angle $\theta_j$ in $P_L$ and with angle $-\theta_j$ in $P_R$ instead of Y rotation to derive pure imaginary amplitudes leading to a negative coefficient.
    If both $y_{j-1}$ and $y_j$ are negative, then one just implements the Y rotation gate to maintain the pure imaginary amplitudes.
\end{proof}
\begin{figure}
    \centering
    \begin{adjustbox}{width=0.6\linewidth}
        \begin{quantikz}
            &\gate{RY(\theta_0)}&\ctrl{1}&\targ{}&&&\ \ldots\ &&\\
            &&\gate{RY(\theta_1)}&\ctrl{-1}&\ctrl{1}&\targ{}&\ \ldots\ &&\\
            &&&&\gate{RY(\theta_2)}&\ctrl{-1}&\ \ldots\ &&\\[0.5cm]
            &&&&\ \ldots\ &\lstick[1,brackets=none,label style={yshift=1cm, xshift=-0.5cm,color=black, rotate=135}]{\ldots}&\ctrl{1}&\targ{}&\\
            &&&&\ \ldots\ &&\gate{RY(\theta_{M-1})}&\ctrl{-1}&\\
        \end{quantikz}
    \end{adjustbox}
    \caption{State-preparation-pair for sparse linear combination of constraints.}
    \label{fig:state_preparation_pair}
\end{figure}
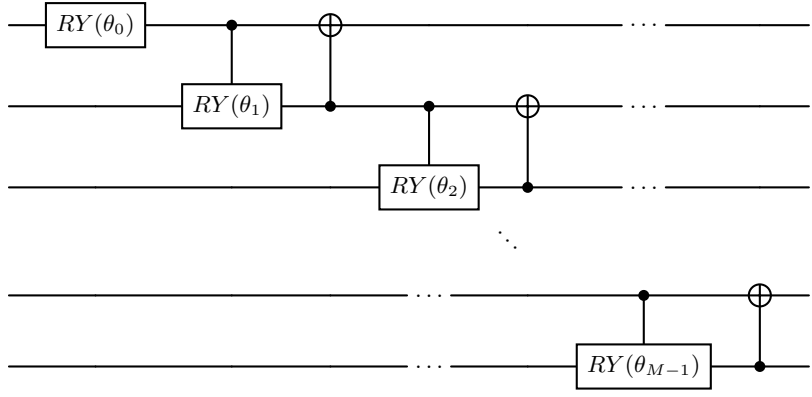

\noindent Armed with these results, one has the following parameterized linear combination of constraints:
\begin{lemma}[\textbf{Sparse linear combination of constraints}]\label{lem:lcc}

    Suppose a linear combination of constraints
    \begin{equation}
        f = \sum_k \lambda_k f_k = \sum_k 2\lambda_k{\lVert f_k\rVert}_{\mathrm{max}}\vec{f_k},
    \end{equation}
    wherein each re-scaled constraint $\vec{f_k}$ can be efficiently prepared via some oracle $O_{f_k}$, a $(1, \log{n}+2, \epsilon_k)$ block-encoding $U_k$ as constructed in Lemma~\ref{lem:nta}.
    Then $f$ can be implemented as
   \begin{equation}
    U_{\mathrm{LCC}} = \sum_{k=0}^M\lambda_k U_k: \ket{0}\rightarrow \frac{1}{\mathcal{N}}\sum_{j=0}^{N-1}(\sum_{k=0}^M\lambda_k f_k(x_j))\ket{j}.
    \end{equation}
    The oracles $O_{f_k}$ are totally queried $M$ times (each $U_k$ is called once) with totally $\mathcal{O}(M\log{n})$ ancillary qubits.
    The total error is $\epsilon_\mathrm{LCC} = 2\lambda\max_k{\epsilon_k}$,
    and the success probability is
    \begin{equation}
        \mathcal{P}_\mathrm{LCC} = \frac{{\lVert f\rVert}_{\mathrm{max}}^2}{4\lambda^2} = \frac{1}{4\vec{\lambda}^2},
    \end{equation}
    where
    \begin{equation}\label{eq:error_sacling_factor}
    \begin{split}
        \lambda =& \sum_{k=1}^{M}|\lambda_k|{\lVert f_k\rVert}_{\mathrm{max}}\text{ and}\\
        \vec{\lambda} =& \sum_{k=1}^{M}|\lambda_k|{\lVert f_k\rVert}_{\mathrm{max}}/{\lVert f\rVert}_{2}
    \end{split}
    \end{equation}
    are the absolute and relative scaling factors, respectively.
\end{lemma}
\begin{proof}
    The linear combination of constraints can be re-written as
    \begin{equation}
        f = \sum_k \lambda_k f_k = \sum_k 2\lambda_k{\lVert f_k\rVert}_{\mathrm{max}}\vec{f_k}
    \end{equation}
    with the coefficients to be $2\lambda_k{\lVert f_k\rVert}_{\mathrm{max}}$.
    According to Lemma~\ref{lem:spp}, there is a $(2\lambda, M-1, \epsilon_\mathrm{spp})$ state-preparation-pair of unitaries with $\lambda$ and $\vec{\lambda}$ the scaling factors defined in Eq.~\eqref{eq:error_sacling_factor}.
    By Lemma~\ref{lem:nta}, the re-scaled constraints $\vec{f_k} = \frac{\lambda_k}{2{\lVert f_k\rVert}_{\mathrm{max}}}f_k$ can be efficiently implemented via a $(1, n+4, \epsilon_f + \epsilon_\mathrm{QSVT1})$ block-encoding.
    Utilizing Lemma~\ref{lem:lcbem} one has a $(2\lambda, n+M+3, \epsilon_\mathrm{spp} + \lambda(\epsilon_f+\epsilon_\mathrm{QSVT1}))$-block-encoding of $f$.
    The success probability on the initial state $\ket{\psi} = H^{\otimes n}\ket{0}$ is lower bounded by    
    \begin{align}
        \mathcal{P}_\mathrm{LCC} &= \frac{\sum_{k=1}^M (2\lambda_k\lVert f_k \rVert_{\mathrm{max}}U_k\ket{\psi})^2}{(\sum_{k=1}^M2|\lambda_k|\lVert f_k \rVert_{\mathrm{max}})^2}\\\label{eq:cauchy}
        &\geq \frac{\frac{1}{M}(\sum_{k=1}^M 2\lambda_k\lVert f_k \rVert_{\mathrm{max}}U_k\ket{\psi})^2}{(\sum_{k=1}^M2|\lambda_k|\lVert f_k \rVert_{\mathrm{max}})^2}\\
        &= \frac{\left\lVert U_\mathrm{LCC}\ket{\psi}\right\rVert^2}{M(\sum_{k=1}^M2|\lambda_k|\lVert f_k \rVert_{\mathrm{max}})^2}\\\label{eq:l2}
        &= \frac{\left\lVert f \right\rVert_2^2}{M(\sum_{k=1}^M2|\lambda_k|\lVert f_k \rVert_{\mathrm{max}})^2}\\
        & =\frac{1}{4M\vec{\lambda}^2}.
    \end{align}
    Herein, Eq.~\eqref{eq:cauchy} is derived by applying Cauchy-Schwartz inequality to the numerator and Eq.~\eqref{eq:l2} is derived from the $l_2$-norm of $f$ on the state $\ket{\psi}$.
    This success probability can be accelerated by amplitude amplification and the query complexity is hence $\mathcal{O}(\sqrt{M}\vec{\lambda})$.
\end{proof}
\noindent\textbf{Remark.} It should be noted that:\\
\textbf{(a)} The success probability is bounded by $\frac{1}{4}$ since
    \begin{equation}
            {\lVert f\rVert}_{\mathrm{max}}
                =\max_x{|\sum_{k=1}^{M}\lambda_k f_k|}
                \leq\max_x{\sum_{k=1}^{M}|\lambda_k||f_k|}
                \leq\lambda.
    \end{equation}
One might argue that the factor $\frac{1}{4}$ is too small, yet it is lower bounded by the QSVT condition $|f|\leq\frac{1}{2}$ (Theorem~56 in Ref.~\citesm{gilyen:2018quantum}) and any related works~\citesm{guo:2021nonlinear, rattew:2023non, conde:2023efficient}.\\
\textbf{(b)} The success probability grows maximally when all the weighted components take their maximum simultaneously. This should not be surprising since, at least intuitively, the largest filling rate is achieved and that is coincident with former works~\citesm{rattew:2022preparing, conde:2023efficient}.\\
\textbf{(c)} There is a one-to-one correspondence between the rotation angles $\vec{\theta}$ and the re-normalized coefficients $2\lambda_k{\lVert f_k\rVert}_{\mathrm{max}}: 1\le k\le M$. Thus it can be viewed as a parameterized layer wherein a quantum optimization method such as \emph{parameter shift rule} can be applied.\\
\textbf{(d)} Herein, each Lagrange multiplier $\lambda_k$ is encoded explicitly as a rotation angle $\theta_k$ in the $U_\mathrm{LCC}$ subroutine (see Figure \ref{fig:medl}\textbf{(b)}).
One noticed that while for common distributions these Lagrange multipliers can be computed from model parameters directly, indeed they can also be learned from some unknown distribution or observed data.
As a result, this SI-PQC can be treated as not only a random variable generator but also a problem-inspired ansatz for quantum optimization algorithms.

\subsection{Maximum Entropy Distribution Loader}\label{sec:medl}
Finally, one can exponentiate the amplitudes encoding of the parameterized constraints combination utilizing Lemma~\ref{lem:nta} once again:
\begin{equation}\label{eq:u_si}
    U_\mathrm{ME}(\vec{\lambda}, \vec{f}):
    \ket{0} \rightarrow
    \frac{1}{\mathcal{N'}}\sum_{j=0}^{N-1}\exp{\sum_{k=1}^M\lambda_k f_k(x_j)}\ket{j}.
\end{equation}
As required in Lemma~\ref{lem:nta}, the implementation of  $U_\mathrm{ME}(\vec{\lambda}, {\vec{f}})$ admits $\mathcal{O}(\log{(\frac{1}{\epsilon})})$ queries of the parameterized constraints $\sum_{k=1}^M\lambda_k U_k$,
and hence can be regarded as a PQC with shared parameters.
In summary, one has:

\noindent\textbf{Proof of Theorem~\ref{thm:med-SI-PQC}:}
\begin{proof}
    \textbf{Correctness:}
    In step 1, by applying the QSVT circuit $U_k$ one can derive the quantum state
    \begin{equation}
        \ket{\vec{f}_k} = \sum_{x\in I}\vec{f}_k(x)\ket{x}.
    \end{equation}
    In step 2, by applying the linear combination of constraints denoted by $U_\mathrm{LCC}$ one has
    \begin{equation}
        \frac{1}{2\lambda}\ket{f} = \sum_{k=1}^M\sum_{x\in I}\frac{\lambda_k\lVert f_k\rVert_{\mathrm{max}}}{2\lambda}\vec{f}_k(x)\ket{x}=\frac{1}{2\lambda}\sum_{x\in I}f(x)\ket{x}.
    \end{equation}
    In step 3, by applying the QSVT circuit $U_{EM}$ one can derive the quantum state
    \begin{equation}
        \frac{1}{2e^{\lVert f\rVert_{\mathrm{max}}}}\sum_{x\in I}e^{f(x)}\ket{x}.
    \end{equation}
    
    \textbf{Error analysis:}
    In step 1, without loss of generality, each constraint function $\vec{f}_k$ is assumed to be approximated by a $d_k$-degree polynomial $P_k$ with a uniform approximation error bound
    \begin{equation}
        |\vec{f}_k-P_k| \leq \epsilon_f (1\leq k\leq M).
    \end{equation}
    And by Lemma~\ref{lem:nta}, each polynomial can be implemented via $U_k$, a $(1, \log{n}+2, \epsilon_\mathrm{QSVT1})$-block-encoding of $P_k$ given the error from the QSVT rotation angles approximation to be $\epsilon_\mathrm{QSVT1}$.
    Thus the error to implement $\vec{f}_k$  is bounded by
    \begin{equation}
        |\vec{f}_k-U_k|\leq |\vec{f}_k-P_k| + |P_k - U_k| \leq \epsilon_f + \epsilon_\mathrm{QSVT1},
    \end{equation}
    and $U_k$ is a $(1, \log{n}+2, \epsilon_f + \epsilon_\mathrm{QSVT1})$-block-encoding of $\vec{f}_k$.
    
    In step 2, one can utilize Lemma~\ref{lem:lcc} to implement the quantum circuit $U_\mathrm{LCC}$, the $(2\lambda, \log{n}+M+5, \epsilon_\mathrm{spp}+\lambda\epsilon_f + \lambda\epsilon_\mathrm{QSVT1})$-block-encoding of $f$ so that
    \begin{equation}
        \left|2\lambda U_\mathrm{LCC}\ket{\psi} - \sum_{x\in I}f(x)\ket{x}\right| \leq \epsilon_\mathrm{spp}+\lambda\epsilon_f + \lambda\epsilon_\mathrm{QSVT1}
    \end{equation}
    holds for the initial state $\ket{\psi} = H^{\otimes n}\ket{0}$ and the interval $I$.
    
    In step 3, one can implement  a quantum circuit $U_\mathrm{ME}$ to be the $(1, \log{n}+M+7, 4d_{\mathrm{exp}}\sqrt{\epsilon_\mathrm{spp}/2\lambda+\epsilon_f + \epsilon_\mathrm{QSVT1}}+\epsilon_\mathrm{QSVT2})$-block-encoding of $P_{\mathrm{exp}}(U_\mathrm{LCC})$, 
    wherein $P_{\mathrm{exp}}$ is the $d_\mathrm{exp}$-degree polynomial approximation of the normalized exponential function $\frac{1}{\mathcal{N}}e^{2\lambda x}$ with the normalization constant
    \begin{equation}\label{eq:normalization_constant}
       \mathcal{N} = \sqrt{\sum_x (e^{2\lambda \frac{1}{2\lambda}f(x)})^2} = \sqrt{\sum_x (e^{f(x)})^2}
    \end{equation}
    and a uniform approximation error bounded by
    \begin{align}\label{eq:taylor}
        \epsilon_\mathrm{exp} &= \left|\frac{e^{2\lambda\frac{f}{2\lambda}}}{\mathcal{N}} - P_e\left(\frac{f}{2\lambda}\right)\right|\\
        &=\left|
        \sum_{j=0}^\infty \frac{(2\lambda)^j}{\mathcal{N}j!} \left(\frac{f}{2\lambda}\right)^j -
        \sum_{j=0}^{d_{\mathrm{exp}}} \frac{(2\lambda)^j}{\mathcal{N}j!} \left(\frac{f}{2\lambda}\right)^j
        \right|\\
        &=\left|
        \sum_{j=d_{\mathrm{exp}}+1}^\infty \frac{(2\lambda)^j}{\mathcal{N}j!} \left(\frac{f}{2\lambda}\right)^j
        \right|\\\label{eq:stirling}
        &\leq \left|
        \sum_{j=d_{\mathrm{exp}}+1}^\infty \frac{\lambda^j}{\mathcal{N}j!}
        \right|
        \leq \left|
        \sum_{j=d_{\mathrm{exp}}+1}^\infty \frac{1}{\mathcal{N}2^j}
        \right|
        \\\label{eq:epsilon_exp_bound}
        &=\frac{1}{\mathcal{N}2^{d_{\mathrm{exp}}}}.
    \end{align}
    Herein, Eq.~\eqref{eq:stirling} holds for any $j\geq2e\lambda$ since by Stirling's formula we have
    \begin{equation}
        \frac{\lambda^j}{j!}\leq\frac{\lambda^j}{\sqrt{2\pi j}(\frac{j}{e})^j}\leq\left(\frac{e\lambda}{j}\right)^j.
    \end{equation}
    And $\epsilon_\mathrm{QSVT2}$ is the error to implement the QSVT of $P_\mathrm{exp}$.
    The error from the imperfect approximation $U_\mathrm{LCC}$ of $f$ can be evaluated as
    \begin{align}\nonumber
        &\left|\frac{1}{\mathcal{N}}e^{2\lambda U_\mathrm{LCC}}\ket{\psi}
        - \sum_{x \in I}\frac{1}{\mathcal{N}}e^{f(x)}\ket{x}
        \right|\\\label{eq:cauchy_mean_value}
        \leq&\frac{1}{\mathcal{N}}e^{\lVert f \rVert_{\mathrm{max}}}
        \left|2\lambda U_\mathrm{LCC}\ket{\psi} - \sum_{x\in I}f(x)\ket{x}
        \right|\\\label{eq:u_lcc_error}
        \leq&\frac{1}{\mathcal{F}}(\epsilon_\mathrm{spp}+\lambda\epsilon_f + \lambda\epsilon_\mathrm{QSVT1}),
    \end{align}
    wherein Eq.~\eqref{eq:cauchy_mean_value} is a result of the Cauchy mean value theorem and Eq.~\eqref{eq:u_lcc_error} is derived from the linear combination of constraints in step 1.
    
    To summarize, the total error to prepare a quantum state proportional to Eq.~\eqref{eq:pdf_me} on the interval $I$ is
    \begin{align}\nonumber
        &\left|U_\mathrm{EM}\ket{\psi} - \sum_{x \in I}\frac{1}{\mathcal{N}}\exp{\sum_{k=1}^M \lambda_k f_k(x)}\ket{x}\right|\\\nonumber
        \leq& \left|U_\mathrm{EM}\ket{\psi} - P_{\mathrm{exp}}(\frac{U_\mathrm{LCC}}{\lambda})\ket{\psi}\right|+\\\nonumber
        &\left|P_{\mathrm{exp}}(\frac{U_\mathrm{LCC}}{\lambda})\ket{\psi}
        -\frac{e^{U_\mathrm{LCC}}}{\mathcal{N}}\ket{\psi}
        \right|+\\\label{eq:triangle}
        &
        \left|\frac{e^{U_\mathrm{LCC}}}{\mathcal{N}}\ket{\psi}
        - \sum_{x \in I}\frac{1}{\mathcal{N}}e^{f(x)}\ket{x}
        \right|\\\nonumber
        \leq&(4d_{\mathrm{exp}}\sqrt{\epsilon_\mathrm{spp}/\lambda+\epsilon_f + \epsilon_\mathrm{QSVT1}}+\epsilon_\mathrm{QSVT2})+\\\label{eq:triangle2}
        &\epsilon_\mathrm{exp}+\frac{1}{\mathcal{F}}(\epsilon_\mathrm{spp}+\lambda\epsilon_f + \lambda\epsilon_\mathrm{QSVT1})\\\label{eq:sqrt_epsilon}
        \leq&(4d_{\mathrm{exp}}+\frac{\lambda}{\mathcal{F}})\sqrt{\epsilon_\mathrm{spp}/\lambda+\epsilon_f + \epsilon_\mathrm{QSVT1}} + \epsilon_\mathrm{exp} +\epsilon_\mathrm{QSVT2}\\\label{eq:rotation_angle}
        \leq&(4d_{\mathrm{exp}}+\frac{\lambda}{\mathcal{F}})\sqrt{\frac{M\Delta\theta_0}{\lambda}+\epsilon_f + d_f\Delta\theta_1} + \epsilon_\mathrm{exp} +d_{\mathrm{exp}}\Delta\theta_2.
    \end{align}
    In  Eq.~\eqref{eq:triangle}, the first part is bounded by the QSVT implementation procedure in step 3, the second part is defined by Eq.~\eqref{eq:taylor}, and the third part is bounded by Eq.~\eqref{eq:u_lcc_error}.
    In Eq.~\eqref{eq:triangle2}, $\epsilon_\mathrm{spp}+\lambda\epsilon_f + \lambda\epsilon_\mathrm{QSVT1}$ is bounded by $\lambda\sqrt{\epsilon_\mathrm{spp}/\lambda+\epsilon_f + \epsilon_\mathrm{QSVT1}}$ since $\epsilon_\mathrm{spp}/\lambda+\epsilon_f + \epsilon_\mathrm{QSVT1}$ is assumed to be less than $1$ without loss of generality.
    In Eq.~\eqref{eq:sqrt_epsilon}, $\epsilon_\mathrm{QSVT1}$ is bounded by $d_{f}\Delta\theta_1$ from Lemma~\ref{lem:qsvt_rotation_error} given the maximum degree of constraints $d_f = \max_{1\leq k\leq M}{d_k}$ and $\epsilon_\mathrm{QSVT2}$ is bounded by $d_{\mathrm{exp}}\Delta\theta_2$ from Lemma~\ref{lem:qsvt_rotation_error}.

    If continuous rotation of an arbitrary rotation angle is promised, then the terms $\Delta\theta_0$, $\Delta\theta_1$ and $\Delta\theta_2$ all vanish and Eq.~\eqref{eq:rotation_angle} can be simplified to be
    \begin{equation}
        (4d_{\mathrm{exp}}+\frac{\lambda}{\mathcal{F}})\sqrt{\epsilon_f}+\epsilon_\mathrm{exp}.
    \end{equation}
    To suppress the error less than a uniform bound $\epsilon$, we observe that $(4d_{\mathrm{exp}}+\frac{\lambda}{\mathcal{F}})\sqrt{\epsilon_f}\leq\epsilon/2$ holds for
    \begin{equation}
        \epsilon_f\leq\frac{\epsilon^2}{4(4d_{\mathrm{exp}}+\frac{\lambda}{\mathcal{F}})^2}.
    \end{equation}
    In practice, the observation functions $f_k$ is usually a polynomial, and hence we have $\epsilon_f=0$.
    By Eq.~\eqref{eq:epsilon_exp_bound}, we also have $\epsilon_\mathrm{exp}\leq\epsilon/2$ holds for
    \begin{equation}
        d_{\mathrm{exp}} \geq \max\{\log\frac{2}{\mathcal{N}\epsilon}, 2e\lambda\}.
    \end{equation}

    If the rotation is imperfect, to suppress the error in Eq.~\eqref{eq:rotation_angle} less than a uniform bound $\epsilon$, one can bound the three terms by $\epsilon/3$ respectively:    
    To bound $(4d_{\mathrm{exp}}+\frac{\lambda}{\mathcal{F}})\sqrt{\frac{M\Delta\theta_0}{\lambda}+\epsilon_f + d_f\Delta\theta_1}$ by $ \epsilon/3$ we have
    \begin{equation}
        \frac{M\Delta\theta_0}{\lambda}+\epsilon_f + d_f\Delta\theta_1 \leq \frac{\epsilon^2}{9(4d_{\mathrm{exp}}+\frac{\lambda}{\mathcal{F}})^2}.
    \end{equation}
    And this holds if
    \begin{align}
        \Delta\theta_0 &\leq \frac{\lambda\epsilon^2}{27M(4d_{\mathrm{exp}}+\frac{\lambda}{\mathcal{F}})^2},\\
        \epsilon_f &\leq \frac{\epsilon^2}{27(4d_{\mathrm{exp}}+\frac{\lambda}{\mathcal{F}})^2},\\
        \Delta\theta_1 &\leq \frac{\epsilon^2}{27d_f(4d_{\mathrm{exp}}+\frac{\lambda}{\mathcal{F}})^2}.
    \end{align}
    By Eq.~\eqref{eq:epsilon_exp_bound}, $\epsilon_\mathrm{exp}\leq\epsilon/3$ holds for
    \begin{equation}
        d_{\mathrm{exp}} \geq \max\{\log\frac{3}{\mathcal{N}\epsilon}, 2e\lambda\}.
    \end{equation}
    And $d_{\mathrm{exp}}\Delta\theta_2\leq\epsilon/3$ holds for
    \begin{equation}
        \Delta\theta_2\leq\frac{\epsilon}{3d_{\mathrm{exp}}}.
    \end{equation}
    
    \textbf{Complexity:}
    The complexity can be directly computed in three steps:
    In step 1, we assume that the $k^{th}$ constraint can be prepared by $U_k$ with circuit depth $c_k$, as guaranteed by Lemma~\ref{lem:nta}.
    In step 2, the linear combination of constraints can be implemented via $U_\mathrm{LCC}$ with circuit depth bounded by
    \begin{equation}
        2(M-1) + \sum_{k=1}^M c_k + 2(M-1) \leq 4M + M\Bar{c} =\mathcal{O}(nMd_f).
    \end{equation}
    Herein, $\Bar{c} = \max_k c_k$ is bounded by $nd_f$ with $d_f = \max_{1\leq k\leq M} d_k$.
    and the success probability is $\mathcal{P}_{\mathrm{LCC}} = \frac{1}{4M\vec{\lambda}^2}$.
    In step 3, one implement $U_\mathrm{ME}$ by
    \begin{equation}
        d_{\mathrm{exp}} < \log\frac{2}{\mathcal{N}\epsilon}+ 2e\lambda < 2e\log{\left(\frac{2e^\lambda}{\mathcal{N}\epsilon}\right)}=2e\log\left(\frac{2}{\mathcal{F}\epsilon}\right)
    \end{equation}
    calls of $U_\mathrm{LCC}$.
    The total circuit depth is hence bounded by
    \begin{equation}
        (4M+M\Bar{c})d_{\mathrm{exp}} = \mathcal{O}\left(nMd_f\log\left(\frac{1}{\mathcal{F}\epsilon}\right)\right).        
    \end{equation}
    And the total success probability is
    \begin{align}
        \mathcal{P}_{ME} 
        &= \left(\left(\sum_{x\in I}\frac{e^{f(x)}}{2e^{\lVert f\rVert_{\mathrm{max}}}}\ket{x}\right)^\dagger\left(\sum_{x\in I}\frac{e^{f(x)}}{{\lVert e^f\rVert_{2}}}\ket{x}\right)\right)^2\\
        &= \left(\frac{\lVert e^f\rVert_{2}}{2\lVert e^f\rVert_{\mathrm{max}}}\right)\\
        &= \frac{\mathcal{F}^2}{4}.
    \end{align}
    Hence the call complexity of $U_\mathrm{ME}$ after amplitude amplification is $\mathcal{O}(\frac{1}{\mathcal{F}})$
    The total time complexity is hence 
    \begin{equation}
        \mathcal{O}\left(\frac{nMd_f}{\mathcal{F}}\log\left(\frac{1}{\mathcal{F}\epsilon}\right)\right).
    \end{equation}
\end{proof}

\subsection{End-to-end Analysis}\label{sec:med-e2e}
Since MEDL is the fundamental subroutine for our further model generation and learning procedures, it is necessary to give a comprehensive analysis of the end-to-end quantum resource analysis.

\begin{corollary}[\textbf{End-to-end analysis of Maximum Entropy Distribution Loader for Common Statistics Distributions}]\label{thm:med2-SI-PQC}
    We summarize the end-to-end preparation complexity for common statistics distributions in Table.~\ref{tab:medl}.
\end{corollary}
\begin{proof}

    \noindent\textbf{(a) Normal distribution with large drift.}\\
    Given $p(x) = \frac{1}{\sqrt{2\pi}\sigma}\exp\left( {-\frac{(x-\mu)^2}{2\sigma^2}} \right)$, one has
    $M=1$, $d_f=2$, and $\mathcal{F}=\frac{\sqrt{\sigma^2+\mu^2}}{1/\sqrt{2\pi}/\sigma}=\sqrt{2\pi(\sigma^2+\mu^2)}\sigma$.
    Then the total complexity is $\mathcal{O}(\frac{1}{\sigma\sqrt{\sigma^2+\mu^2}}n\log{\frac{1}{{\sigma\sqrt{\sigma^2+\mu^2}}\epsilon}})$.

    \noindent\textbf{(b) Exponential distribution.}\\
    Given $p(x)=\lambda \exp(-\lambda x)$, it is easy to check that $M=1$, $d_f=\log{(1/\epsilon)}\frac{\ln{(1-c)}}{\ln{(1-c/2)}}$ and
    $\mathcal{F}=\frac{\sqrt{2}/\lambda}{\lambda}=\frac{\sqrt{2}}{\lambda^2}$.
    Consequently, the preparation complexity is $\mathcal{O}(\frac{n}{\lambda^2}\log{(\lambda/\epsilon)})$.

    \noindent\textbf{(c) Pareto distribution.}\\
    Given $p(x)=\frac{\alpha x_0^\alpha}{x^{\alpha+1}}$,
    one needs to prepare the constraint $\ln{x}$.
    Observed that $\ln{x} = \ln(1+c(\frac{x}{c}-\frac{1}{c}))$,
    this constraint can be implemented through the nonlinear transformation of $\ln(1+cx)$ in Lemma~\ref{lem:nta} 
    as a $d=\log(1/\epsilon)$ polynomial with filling rate $\mathcal{F}=\frac{\ln(1-c/2)}{\ln(1-c)}$.
    with $\mathcal{O}({\log(1/\epsilon)\frac{\ln(1-c)}{\ln(1-c/2)}})$ queries of the general linear function loader in Lemma~\ref{lem:lft}.
    One also has $M=1$ and $\vec{\lambda}=1$.
    Consequently, the total complexity is $\mathcal{O}(n\log^2(\frac{1}{\epsilon}))$.

    \noindent\textbf{(d) Rayleigh distribution.}\\
    Given $p(x)=\frac{x}{\sigma^2}\exp\left( {-\frac{x^2}{2\sigma^2}} \right)$ on the interval $[0, x_0]$, one has $p(x) = \frac{e^{-\frac{1}{2\sigma^2}}}{\sigma^2}\exp{\ln{(1+cu)} -\frac{c^2}{2\sigma^2}u^2 - \frac{c}{\sigma^2}u}$
    wherein $u = (\frac{x}{c}-\frac{1}{c})$ is a linear function.
    One has $\vec{\lambda} = (1, -\frac{c^2}{2\sigma^2}, -\frac{c}{\sigma^2})$, $\vec{f} = (\ln{(1+cu)}, u^2, u)$,
    and a similar computation shows that $M = 3$, $d_f=\log{(1/\epsilon)}\frac{\ln{(1-c)}}{\ln{(1-c/2)}}$ and
    \begin{equation}
            \mathcal{F} = \frac{\sqrt{\frac{1}{2\sigma^2x_0}\gamma(2, \frac{x_0^2}{\sigma^2})}}{1/(\sigma\sqrt{e})}\\
            =\sqrt{\frac{e}{2x_0}\gamma(2, \frac{x_0^2}{\sigma^2})},
    \end{equation}
    where $\gamma(2, t)=\int_0^tue^{-u}du$ is the lower incomplete gamma function.
    Hence the circuit depth is 
    \begin{equation}
        \mathcal{O}\left(\frac{\sqrt{x_0}n}{\gamma\left(2, \frac{x_0^2}{\sigma^2}\right)}\log^2{\left(\frac{x_0}{\gamma\left(2, \frac{x_0^2}{\sigma^2}\right)\epsilon}\right)}\right).
    \end{equation}

    \noindent\textbf{(e) Chi distribution.}\\
    Given $p(x)=\frac{2}{2^{k/2}\Gamma(k/2)}x^{k-1}\exp\left( {-\frac{x^2}{2}} \right)$, one has
    \begin{equation}
        \begin{split}
            p(x) &\propto \exp\left((k-1)\ln{x}-\frac{x^2}{2}\right)\\
            &\propto \exp\left((k-1)\ln{(1+cu)}-\frac{c^2}{2}u^2-cu\right),
        \end{split}
    \end{equation}
    wherein $u = (\frac{x}{c}-\frac{1}{c})$ is a linear function.
    One has $\vec{\lambda} = (k-1, -\frac{c^2}{2}, -c)$, $\vec{f} = (\ln{(1+cu)}, u^2, u)$,
    and a similar computation shows that $M = 3$, $d_f=\log{(1/\epsilon)}\frac{\ln{(1-c)}}{\ln{(1-c/2)}}$ and
    \begin{equation}
            \mathcal{F} = \frac{\sqrt{k}}{\sqrt{2}(\frac{k-1}{2e})^{\frac{k-1}{2}}/\Gamma(\frac{k}{2})}\\
            =\frac{\sqrt{k}\Gamma(\frac{k}{2})}{\sqrt{2}(\frac{k-1}{2e})^{\frac{k-1}{2}}},
    \end{equation}
    where $\Gamma(z)=\int_0^\infty t^{z-1}e^{-t}dt$ is the Gamma function.
    Hence the circuit depth is $\mathcal{O}\left(\frac{n}{\mathcal{F}}\text{polylog}(\frac{1}{\mathcal{F}}, \frac{1}{\epsilon})\right)$.

    \noindent\textbf{(f) Chi-squared distribution.}\\
    Given $p(x)=\frac{1}{2^{k/2}\Gamma(k/2)}x^{k/2-1}\exp\left( {-\frac{x}{2}} \right)$, one has
    \begin{equation}
        \begin{split}
            p(x) &\propto \exp\left((k/2-1)\ln{x}-\frac{x}{2}\right)\\
            &\propto \exp\left((k/2-1)\ln{(1+cu)}-\frac{c}{2}u\right),
        \end{split}
    \end{equation}
    wherein $u = (\frac{x}{c}-\frac{1}{c})$ is a linear function.
    One has $\vec{\lambda} = (k-1, -\frac{c}{2})$, $\vec{f} = (\ln{(1+cu)}, u)$,
    and a similar computation shows that $M = 2$, $d_f=\log{(1/\epsilon)}\frac{\ln{(1-c)}}{\ln{(1-c/2)}}$ and
    \begin{equation}
            \mathcal{F} = \frac{\sqrt{k(k+2)}}{(\frac{k-2}{2e})^{k/2-1}/2/\Gamma(\frac{k}{2})}\\
            =\frac{2\sqrt{k(k+2)}\Gamma(\frac{k}{2})}{\sqrt{2}(\frac{k-1}{2e})^{k/2-1}}.
    \end{equation}
    Hence the circuit depth is $\mathcal{O}\left(\frac{n}{\mathcal{F}}\text{polylog}(\frac{1}{\mathcal{F}}, \frac{1}{\epsilon})\right)$.
\end{proof}
\noindent\textbf{Remark.}
Firstly, the preparation complexity for normal distribution with no drift herein is indeed coincident with Theorem~5 of Ref.~\citesm{rattew:2023non} once observed that the assumption of $-1\leq x\leq1$ induces the restriction $2\sigma^2\geq1$.
Moreover, the complexity dependency on $\sigma$ can be removed as large $\sigma$ leads to an easily-preparing flat state.
Secondly, the preparation complexity for normal distribution with small drift herein depends on the relative variance $\mu/\sigma$.
Intuitively, the larger the variance, the state is closer to a uniform distribution that is easy to prepare.

\section{Weighted Mixture of Distributions}\label{sec:wdm-SI-PQC}
In this subsection, we consider the more general case of a latent distribution of many visible distributions,
wherein different visible distributions are assigned with corresponding probabilities from the latent distribution of their types and parameters.
More formally one has:
\begin{definition}[\textbf{Weighted mixture of distributions~\citesm{teicher:1960on}}]\label{def:wmd}
    
    If $\Omega = \{ f \}$ is a family of distribution functions and $\omega$ is a measure on a Borel Field of subsets of $\Omega$ with $\omega(\Omega) = 1$ (intuitively, $\omega$ is the probability measure on the space of many visible distributions $\Omega$).
    Then
    \begin{equation}\label{eq:wmd}
        \int f(\cdot)d\omega(f)
    \end{equation}
     is again a distribution function which is called a $\omega$-mixture of $\Omega$.
\end{definition}
\noindent This issue is quite common for statistics modeling, especially for practical data science and financial problems, since the types and parameters of the underlying distribution are usually not sure.

\subparagraph{Distinct distributions mixer.}
Noted the similarity between the random distribution specification for weighted mixture and the state collapse for quantum superposition,
we realize a natural quantum circuit implementation for the weighted distribution mixer Eq.~\eqref{eq:wmd} as a linear combination of distributions.
More specifically, suppose $M$ distinct distributions $f_k$ that can be prepared via some oracle $O_k$ and $M$ positive real-valued numbers of weights $w_k$ $(1\leq k\leq M)$ such that $\sum_{k=1}^Mw_k=1$.
Then the weighted mixture p.d.f.
\begin{equation}
    f(x) = \sum_{k=1}^Mw_kf_k(x)
\end{equation}
can be efficiently implemented in the same way as in Lemma~\ref{lem:lcc}.
One should notice that no assumption on the oracle $O_k$ is made,
and the underlying distributions can be prepared in a flexible way, including MPS, WSL, QSVT, quantum generative adversarial networks (QGAN), and many other shallow quantum circuits~\citesm{Grover:2002creating, holmes:2020efficient, zylberman:2023efficient, rattew:2022preparing, conde:2023efficient, zoufal:2019quantum}.
Importantly, both the hidden distribution weights and the underlying distribution parameters are exposed to follow our \emph{Static-Structured, Tunable-Parameterized and Modular Paradigm}.

\subparagraph{Parametric family mixer.}
For practical problems, it is often more interesting to tackle with a parametric family mixture:
for example, the time interval between random events is usually supposed to be an exponential distribution mixture with different decay rates,
while the financial instrument price is often assumed to be a normal distributions mixture with varying means and variances.
More formally, we define:
\begin{definition}[\textbf{Distribution Mixture of a Parametric Family}]\label{def:dmpf}
    
    Suppose an $M$-dimensional latent parameter space $\Omega_\theta = \{(\theta_1, ..., \theta_M)\}$, the latent distribution with p.d.f. to be $p(\theta)$, and a parametric family of underlying distributions defined on the visible data $x$ and parameterized by $\theta$ with p.d.f. to be $f(x; \theta)$.
    Then the $p$-weighted mixture is
    \begin{equation}\label{eq:dmpf}
        F = \int f(x; \theta) p(\theta) d\theta.
    \end{equation}
\end{definition}
\noindent Note that Eq.~\eqref{eq:dmpf} is an integration over the parameter space, 
and no previous works can prepare the corresponding quantum state directly and easily up to known.
Moreover, a direct application of the LCU method, as mentioned above, will lead to unaffordable resource consumption:
The circuit depth and the qubit number both grow linearly with the whole parameter space size $|\Omega_\theta|\propto(1/\Delta\theta)^M$, hence polynomially quick with the precision $\epsilon\propto1/\Delta\theta$ and exponentially quick with the free parameter number $M$.

Due to our \emph{Static-Structured, Tunable-Parameterized and Modular Paradigm}, the SI-PQC proposed in the previous subsection can be utilized to prepare the desired state with exponential improvement as follows:

\noindent\textbf{Step 1. Prepare latent distribution on parameter space.}
Firstly, we prepare the latent distribution of parameters.
Without loss of generality, one may assume $M$ independent parameters $\theta_k (1\leq k\leq M)$ defined on $k^{th}$ register of $n_k^{(\theta)} = \log{\frac{1}{\Delta\theta_k}}$ qubits.
Each parameter $\theta_k$ follows a latent distribution with p.d.f. $p_k^{(\theta)}(j)$ that can be prepared to a quantum superposition via an operator $L_k$ with circuit depth $d_k^{(\theta)}$:
\begin{equation}
    L_k\ket{0}_{n_k^{(\theta)}} = \sum_{j=1}^{2^{n_k^{(\theta)}}-1} \sqrt{p_k^{(\theta)}(j)}\ket{j}.
\end{equation}
Then the latent distribution on the whole parameter space can be prepared as
\begin{equation}
    \begin{split}
        &(L_1\otimes\hdots\otimes L_M)\ket{0}_{n_{\theta}}\\
        =& \bigotimes_{k=1}^M \left( \sum_{j=1}^{2^{n_k^{(\theta)}}-1} \sqrt{p_k^{(\theta)}(j)}\ket{j}\right)\\
        =& \sum_\theta \sqrt{p^{(\theta)}(\theta_1, ..., \theta_M)} \ket{\theta_1, ..., \theta_M}_{n_{\theta}}
    \end{split}    
\end{equation}
on $n_{\theta} = \sum_{k=1}^M n_k^{(\theta)}$ qubits with circuit depth $d_{\theta} = \max_k d_k^{(\theta)}$.
Notice that we create a superposition of all possible parameters,
wherein the parameters are digital-encoded and the latent distribution is stored in the probability amplitudes.

\noindent\textbf{Step 2. Digital-analog parameters conversion.}
Secondly, we implement a digital-analog conversion before injecting the parameter information into the quantum circuit.
Remember that we prepare digital-encoded parameters in Step 1,
while for \emph{maximum entropy distribution loader} parameters should be encoded in the amplitudes, as required in the linear combination subroutine (see Lemma~\ref{lem:lcc}, and also Lemma~\ref{lem:spp} for more details).
One can implement this conversion by a sequence of controlled rotation gates which are uniformly controlled by the superposition in the digital parameter register $\mathrm{DP}$ and generate the corresponding amplitudes in the analog parameter register $\mathrm{AP}$:
\begin{equation}
    \sum_\theta \sqrt{p^{(\theta)}(\theta)} \ket{\theta}_\mathrm{DP}(\sum_{k=1}^M \mathrm{Amp}(\theta, k)\ket{k}_\mathrm{AP}).
\end{equation}

\noindent\textbf{Step 3. Prepare visible distribution on data space.}
Thirdly, since the amplitude-encoded parameters have been efficiently prepared one can prepare the visible distribution on data space via SI-PQC.
Therein the visible distribution is generated through the fixed structure circuit, parameterized by the analog-encoded parameters sampled from the latent distribution as desired.
To summarize, we have\\
\noindent\textbf{Proof of Theorem~\ref{thm:wdm-SI-PQC}:}
\begin{proof}    
    In Step 1, as shown in Fig.~\ref{fig:state_preparation_for_wdm}{(a)}, the digital-encoded parameters are prepared on the first $M$ registers of $n^{(\theta)}$ qubits with circuit depth
    \begin{equation}
        d_{\theta} = \max_k d^{(\theta)}_k.
    \end{equation}
    
    In Step 2,  to satisfy the requirement in Lemma~\ref{lem:lcc},
    one need to construct a state-preparation-pair controlled by the digital-encoded parameters.
    This is implemented through a sequence of controlled rotation gates as follows (See Fig.~\ref{fig:state_preparation_for_wdm}{(b)} for reference):
    Assume the digital-encoded parameter to be
    \begin{equation}
        \theta_1 = \theta_0\times0.s_1^{(1)}s_2^{(1)}...s_{n_1^{(\theta)}}^{(1)},
    \end{equation}
    wherein $\theta_0$ is a common bound for all parameters $\theta_k < \theta_0, 1\leq k\leq M$ and $0.s_1^{(1)}s_2^{(1)}...s_{n_1^{(\theta)}}^{(1)}$ is the related binary representation of $\theta_1/\theta_0$.
    Then it can be converted into an analogue-encoded state
    \begin{equation}
        \cos{\frac{\theta_1}{2}}\ket{0} + \sin{\frac{\theta_1}{2}}\ket{1}
    \end{equation}
    through $n_1^{(\theta)}$ controlled rotation gates with control qubit $s_k^{(1)}$ and target qubit on the analogue-encoded parameters register.
    Repeating this subroutine $M$ times one can derive the state-preparation-pair
    \begin{equation}
    \begin{aligned}        
        \cos{\frac{\theta_1}{2}}\ket{00..0} &+ \sin{\frac{\theta_1}{2}}\cos{\frac{\theta_2}{2}}\ket{10..0} + ... \\        
        &+\sin{\frac{\theta_1}{2}}...\sin{\frac{\theta_M}{2}}\ket{11..1},
    \end{aligned}
    \end{equation}
    with circuit depth
    \begin{equation}
        \sum_{k=1}^M n_k^{(\theta)} = n_{\theta}.
    \end{equation}

    In Step 3, to apply the linear combination of distribution, one needs to apply the above subroutine twice as well as controlled-preparation of the visible distribution.
    Consequently, the additional depth to implement the mixture is $\mathcal{O}((d_\theta+n_\theta)\log{\frac{1}{\epsilon\mathcal{F}}})$
    The total circuit depth is hence bounded by $\mathcal{O}((\frac{nd_xM}{\mathcal{F}}+d_\theta+n_\theta)\log{\frac{1}{\epsilon\mathcal{F}}})$.
\end{proof}
\begin{figure*}
\begin{adjustbox}{width=0.85\linewidth}
    \begin{quantikz}
        \lstick[3,brackets=right,label style={rotate=90, yshift=0.5cm, xshift=1cm}]{{Digital-encoded\\ parameters}}&\qwbundle{n_1^{(\theta)}}
        \lstick[1,brackets=none,label style={yshift=0.5cm, xshift=-1cm,color=black}]{\textbf{(a)}} &\gate{L_1}&\lstick[1,brackets=none,label style={ yshift=0.5cm, xshift=3.5cm,color=blue}]{$\ket{\theta_1}=\sum_{j=0}^{2^{n_1^{(\theta)}}}\sqrt{p_1^{(\theta)}(j)}\ket{j}$}&&&&&&&&\ctrl{3}&&&\ \ldots\ &&\gate{L_1^\dagger}&\\ 
        &\qwbundle{n_2^{(\theta)}} &\gate{L_2}&\lstick[1,brackets=none,label style={ yshift=0.5cm, xshift=3.5cm,color=blue}]{$\ket{\theta_2}=\sum_{j=0}^{2^{n_2^{(\theta)}}}\sqrt{p_2^{(\theta)}(j)}\ket{j}$}&&&&&&&&&\ctrl{3}&&\ \ldots\ &&\gate{L_2^\dagger}&\\[1cm]
        &\qwbundle{n_M^{(\theta)}}\lstick[1,brackets=none,label style={yshift=1.5cm, xshift=0.9cm,color=black, rotate=90}]{\ldots} &\gate{L_M}&\lstick[1,brackets=none,label style={ yshift=0.5cm, xshift=3.75cm,color=blue}]{$\ket{\theta_M}=\sum_{j=0}^{2^{n_M^{(\theta)}}}\sqrt{p_M^{(\theta)}(j)}\ket{j}$}&&&&&&&&&&&\ \ldots\ &\ctrl{4}&\gate{L_M^\dagger}\lstick[1,brackets=none,label style={yshift=1.5cm, xshift=0.4cm,color=black, rotate=90}]{\ldots}&\\
        \lstick[4,brackets=right,label style={rotate=90, yshift=0.5cm, xshift=1cm}]{{Analogue-encoded\\ parameters}}&&&&&&&&&&&\gate{RY(\theta_1)}&\ctrl{1}&\targ{}&\ \ldots\ &&&\\
        &&&&&&&&&&&&\gate{RY(\theta_2)}&\ctrl{-1}&\ \ldots\ &&&\\[1cm]
        &\lstick[1,brackets=none,label style={yshift=1.2cm, xshift=0.9cm,color=black, rotate=90}]{\ldots}&&&&&&&&&&&&\lstick[1,brackets=none,label style={yshift=1.2cm, xshift=0.8cm,color=black, rotate=135}]{\ldots}&\ \ldots\ &\ctrl{1}&\targ{}&\\
        &&&&&&&&&&&&&&\ \ldots\ &\gate{RY(\theta_M)}&\ctrl{-1}&\\
    \end{quantikz}
\end{adjustbox}
\begin{adjustbox}{width=0.85\linewidth}
    \begin{quantikz}
        &\qwbundle{n_k^{(\theta)}} 
        \lstick[1,brackets=none,label style={yshift=0.5cm, xshift=-1cm,color=black}]{\textbf{(b)}}
        &\gate{L_k}&\ctrl{1}&\\[2.3cm]
        &&&\gate{RY(\theta_k)}&\\
    \end{quantikz}
    =
    \begin{quantikz}
        &\gate[3]{L_k}&\ctrl{3}&&&&\\
        &&&\ctrl{2}&&&\\
        &&&&\ \ldots\ &\ctrl{1}&\\[0.6cm]
        &&\gate{RY(\theta_0/2^1)}&\gate{RY(\theta_0/2^2)}&\ \ldots\ &\gate{RY(\theta_0/2^{n_k^{(\theta)}})}&\\
    \end{quantikz}
\end{adjustbox}
    \caption{
    Circuit implementation for the weighted distribution mixer of a parametric family.
    (a)
    The quantum circuit converts the digital-encoded latent space into a state preparation pair on the analog-encoded space.
    (b)
    The digital-analog conversion for each parameter.
    }
    \label{fig:state_preparation_for_wdm}
\end{figure*}
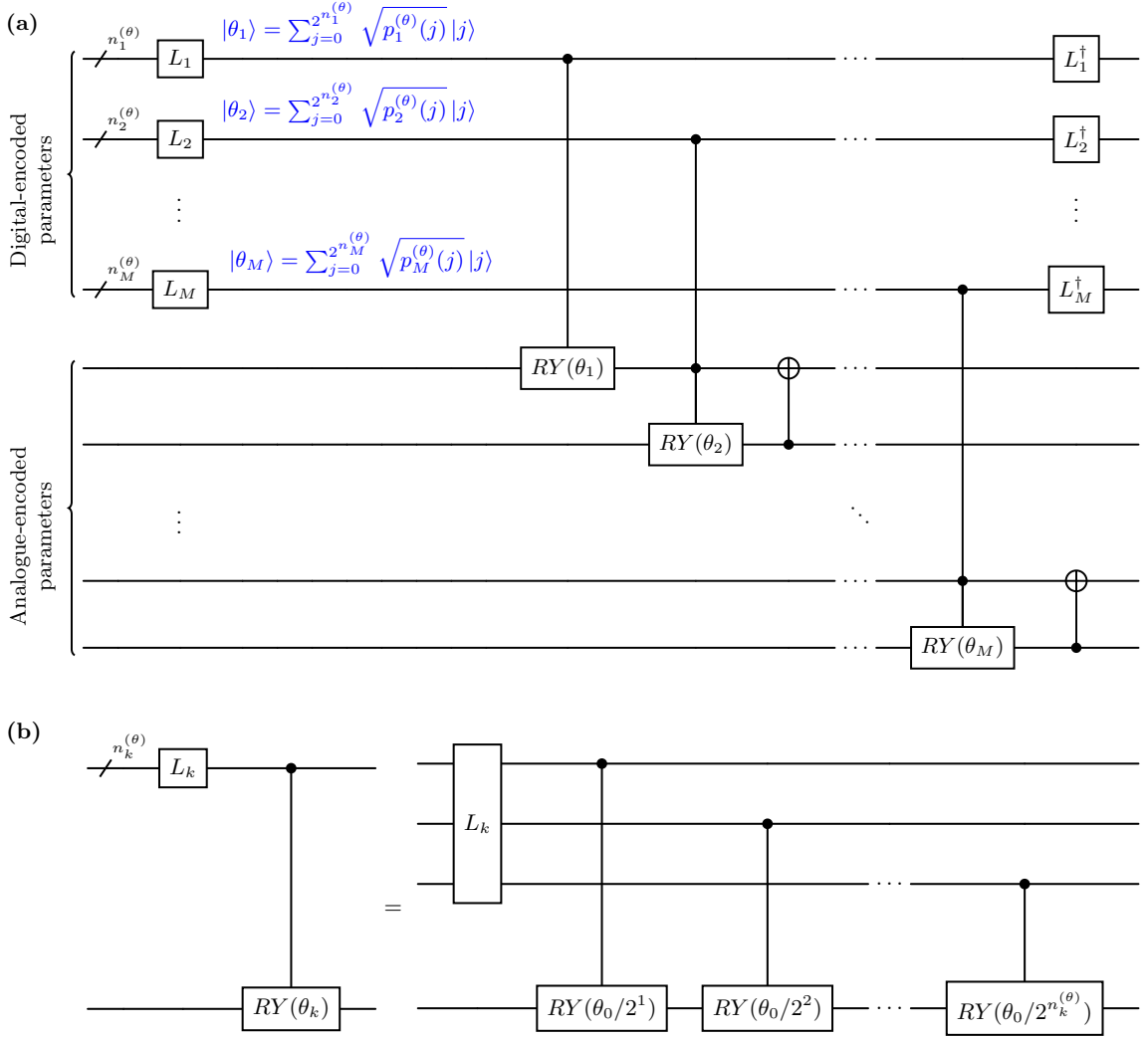
\noindent\textbf{Remark.} Notice that $n^{(\theta)} = \sum_{k=1}^M \log\frac{1}{\Delta\theta_k} = \log\frac{1}{\Delta\theta
}$ with $\Delta\theta$ is the grind size of parameter space, 
and hence this \emph{weighted distribution mixer} SI-PQC makes an exponential improvement on the parameter space precision.

\subparagraph{Gaussian Mixture Model.}
To illustrate the applicability of weighted distribution mixer, we consider the quantum state preparation for the celebrated Gaussian mixture model with applications in a variety of fields such as finance and biometrics~\citesm{douglas:2009gaussian, ian:2008portfolio}.
Specifically, a Gaussian mixture is a linear combination of different normal distributions with varying means and variances,
and it can be modeled through our weighted distribution mixer SI-PQC as
\begin{corollary}\label{thm:gm}
    \textbf{(Gaussian Mixture).} 
    Suppose a Gaussian mixture of $(N_\mu N_\sigma)$ normal distributions with $N_\mu$ means bounded by $|\mu_k|\leq\mu<1$ and $N_\sigma$ variances bounded by $\sigma_k\geq\sigma$.
    Further, assume that the desired parameters $(N_\mu$ and $N_\sigma)$ are generated by a parameterized quantum circuit of depth $d_{\theta}$.
    Then it can be prepared on $n$ qubits through a quantum circuit of depth $\mathcal{O}((n_x+d_{\theta}+n_{\theta})\log{\frac{1}{\sigma\sqrt{\sigma^2+\mu^2}\epsilon}})$ and total complexity $\mathcal{O}(\frac{n_x+d_\theta+n_\theta}{\sigma\sqrt{\sigma^2+\mu^2}}\log{\frac{1}{\sigma\sqrt{\sigma^2+\mu^2}\epsilon}})$.
\end{corollary}
\begin{proof}
    The qubit number of latent distribution space is
    \begin{equation}
        n_{\theta} = \log{N_\sigma} + \log{N_\mu} = \log{N_\sigma N_\mu}.
    \end{equation}
    Consequently, the circuit complexity is 
    \begin{equation}
        \mathcal{O}(\frac{n_x+d_\theta+n_\theta}{\sigma\sqrt{\sigma^2+\mu^2}}\log{\frac{1}{\sigma\sqrt{\sigma^2+\mu^2}\epsilon}}).
    \end{equation}
\end{proof}
\noindent\textbf{Remark.}
The complexity dependency on the latent space preparation depth $d^{(\theta)}$ characterizes the expressiveness of learning the hidden distribution of parameters.
The complexity dependency on the sampling space size is exponentially reduced to $\log{N_\mu N_\sigma}$.

\section{Quantum-enhanced Statistics Model Learning}\label{sec:learning}
In this section, we establish the analysis framework based on SI-PQC to illustrate its applicability in solving essential statistics tasks.

Here, we consider the fundamental statistical task of model calibration to minimize the distance/similarity between the empirical observed data from the realistic world and the epistemic predicted data from statistics models.
Note that model calibration is the standard procedure when applying any statistics model to the empirical data,
and is preferred, benefiting from its theoretical interpretability and practical prior knowledge, in a variety of fields including bio-statistics, drug design, actuarial science, elasticity imaging, source identification, and damage detection~\citesm{warner:2015stochastic}.
Formally, in the model calibration setup, one is supposed to optimize:
\begin{equation}
    \argmin_{\theta} f(Y_\mathrm{obs}, Y_\mathrm{pre}(x|\theta)),
\end{equation}
wherein $f$ denotes the objective function of the optimization problem known as the \textit{calibration metric}, $Y_\mathrm{obs}$ is the observed data,  $Y_\mathrm{pre}$ is the predicted data and $\theta$ are the parameters to be optimized.

However, classical statistics model calibration is facing challenges as theoretical resolution of explicit form is infeasible and numerical method based on Monte-Carlo simulation will consume a huge amount of computation resources and time~\citesm{warner:2015stochastic, lee:2019review}.
To address this issue, We show that the quantum computational statistics models based on SI-PQC enable a feasible and flexible framework to implement a model calibration with potential quantum advantage:

\noindent\textbf{Step 1. Prepare the observed and predicted data.}
To implement the preparation of the observed data, we make no assumption or restriction considering the flexibility of our framework: 
it can be accessed from either an oracle of quantum data or an efficient superposition preparation of classical data,
and it can be either a probability density function or a cumulative density function of the observed data.
For example, we can assume here access to the empirical data histogram as
\begin{equation}\label{eq:quantum_data_oracle}
    \ket{P_\mathrm{obs}} = O_\mathrm{obs}\ket{0}_n = \sum_{j=1}^{2^n}p(j)\ket{j},
\end{equation}
wherein $p(j)=\sqrt{\mathbb{P}_\mathrm{obs}(j)}$ is the probability amplitude of the $j^{th}$ bin of the observed data histogram.
As for the preparation of the predicted data, we can choose any family of SI-PQC mentioned above specified for the problem
\begin{equation}\label{eq:SI-PQC_oracle}
    \ket{P_\mathrm{pre}} = U_\mathrm{SI-PQC}\ket{0}_n = \sum_{j=1}^{2^n}\hat{p}(j)\ket{j},
\end{equation}
wherein $\hat{p}(j)=\sqrt{\mathbb{P}_\mathrm{pre}(j)}$ is the probability amplitude of the $j^{th}$ bin of the predicted data distribution.
Note that the statistics model specification from prior knowledge is one most significant differences from 
model calibration and general machine learning,
this is usually not a limitation in principle yet an empirical information injection in practice.

\noindent\textbf{Step 2. Compute the calibration metric.}
To compute the calibration metric, there are many types of distance/similarity objective functions to be suggested by statisticians, including Euclidean/Minkowski family $L_p = (\sum(P_\mathrm{obs}, P_\mathrm{pre})^p)^{1/p}$, inner product family $\sum P_\mathrm{obs}P_\mathrm{pre}$ and fidelity family $\sum \sqrt{P_\mathrm{obs}P_\mathrm{pre}}$ as discussed in Ref.~\citesm{cha:2007comprehensive}.
This type of computation can be efficiently implemented on a quantum processor by introducing a SWAP-test, a Hadamard test, or any other quantum state distance metric.
For convenience, here we consider the state fidelity between the observed and predicted states as
\begin{equation}
    \braket{P_\mathrm{pre}}{P_\mathrm{obs}} = \bra{0}U_\mathrm{SI-PQC}^\dagger O_\mathrm{obs}\ket{0}.
\end{equation}
It should be mentioned that the operator $U_\mathrm{SI-PQC}^\dagger$ can be efficiently implemented as a direct consequence of our explicit construction of the quantum statistics model.
Further, the observed data state $\ket{P_\mathrm{obs}}$ can be regarded as either from preparation or a black-box oracle with no prior knowledge, enlarging the scope of applications.

\noindent\textbf{Step 3. Learn the model parameters.}
To optimize the model parameters, we apply a quantum-classical hybrid optimization procedure wherein a quantum processor is promised to evaluate the objective function efficiently as discussed above and a classical optimizer is employed to update the parameters~\citesm{cerezo:2021variational}.
As these statistics parameters are encoded in the rotation angles as guaranteed by our SI-PQC protocol, variational quantum algorithm-typed techniques such as parameter shift rule are enabled.

\section{Additional Numerical Results}\label{app:additional_numerical}

In this section, we provide additional numerical results for preparing and training various mixtures of distributions to showcase the generality of SI-PQC.

\subsection{Preparing Mixtures}\label{app:additional_numerical_prepare}

We prepare a weighted mixture of eight different distributions from the weight vectors listed in Tab.~\ref{tab:exp_mix} for exponential distributions and normal distributions.

Fig.~\ref{fig:exp-mixture} shows the numerical results to prepare an $8$-mixture of exponential distributions with varying rates $\vec{\lambda} = (0.2, 0.4, 0.6, 0.8, 2, 4, 6, 8)$. The experimental results (shown in various markers) agree with the theoretical benchmarks (shown in lines).

Fig.~\ref{fig:gaussian-mixture} shows the numerical results to prepare an $8$-mixture of normal distributions with varying means $\vec{\mu} = (1, 3, 5, 7, 9, 11, 13, 15)$ and variances $\vec{\sigma^2} = (25, 30.25, 36, 42.25, 49, 56.25, 64, 72.25)$. The experimental results (shown in various markers) also agree well with the theoretical benchmarks (shown in lines).
\begin{table}[htp]
    \caption{Weight vectors in $8$-mixture experiments.}
    \centering    
    \begin{tabular}{ c c }
    \toprule[1pt]
        Weights vector & Value\\
        \midrule[0.3pt]
             $\vec{w}_1$ & $(\frac{1}{36}, \frac{2}{36}, \frac{3}{36}, \frac{4}{36}, \frac{5}{36}, \frac{6}{36}, \frac{7}{36}, \frac{8}{36})$ \\ 
             $\vec{w}_2$ & $(\frac{2}{36}, \frac{3}{36}, \frac{4}{36}, \frac{5}{36}, \frac{6}{36}, \frac{7}{36}, \frac{8}{36}, \frac{1}{36})$ \\  
             $\vec{w}_3$ & $(\frac{3}{36}, \frac{4}{36}, \frac{5}{36}, \frac{6}{36}, \frac{7}{36}, \frac{8}{36}, \frac{1}{36}, \frac{2}{36})$    \\  
             $\vec{w}_4$ & $(\frac{4}{36}, \frac{5}{36}, \frac{6}{36}, \frac{7}{36}, \frac{8}{36}, \frac{1}{36}, \frac{2}{36}, \frac{3}{36})$    \\  
             $\vec{w}_5$ & $(\frac{5}{36}, \frac{6}{36}, \frac{7}{36}, \frac{8}{36}, \frac{1}{36}, \frac{2}{36}, \frac{3}{36}, \frac{4}{36})$    \\  
             $\vec{w}_6$ & $(\frac{6}{36}, \frac{7}{36}, \frac{8}{36}, \frac{1}{36}, \frac{2}{36}, \frac{3}{36}, \frac{4}{36}, \frac{5}{36})$    \\  
             $\vec{w}_7$ & $(\frac{7}{36}, \frac{8}{36}, \frac{1}{36}, \frac{2}{36}, \frac{3}{36}, \frac{4}{36}, \frac{5}{36}, \frac{6}{36})$    \\  
             $\vec{w}_8$ & $(\frac{8}{36}, \frac{1}{36}, \frac{2}{36}, \frac{3}{36}, \frac{4}{36}, \frac{5}{36}, \frac{6}{36}, \frac{7}{36})$  \\  
     \bottomrule[1.5pt]
    \end{tabular}
    \label{tab:exp_mix}
\end{table}
\begin{figure}[hp]
    \centering
    \includegraphics[width=0.6\linewidth]{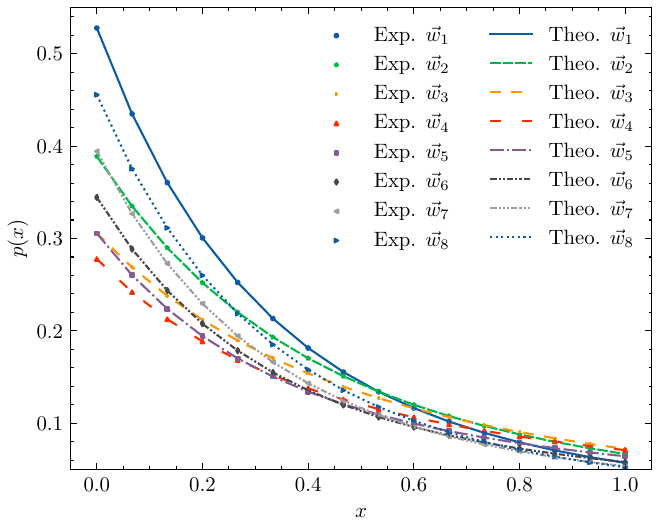}
    \raggedleft
    \centering
    \caption{The numerical results to prepare a mixture of eight exponential distributions with rates $\vec{\lambda}=(0.2, 0.4, 0.6, 0.8, 2, 4, 6, 8)$, 
    and weights list in Tab.~\ref{tab:exp_mix}.
    }
    \label{fig:exp-mixture}
\end{figure}
\begin{figure}[hp]
    \centering
    \includegraphics[width=0.5\linewidth]{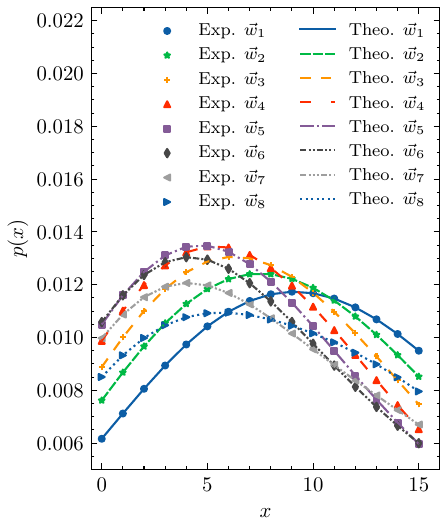}
    \raggedleft
    \centering
    \caption{The numerical results to prepare a mixture of eight normal distributions with means $\vec{\mu}=(1, 3, 5, 7, 9, 11, 13, 15)$, variances $\vec{\sigma^2}=(25, 30.25, 36, 42.25, 49, 56.25, 64, 72.25)$
    and weights list in Tab.~\ref{tab:exp_mix}.
    }
    \label{fig:gaussian-mixture}
\end{figure}

\subsection{Training Mixtures}\label{app:additional_numerical_train}

Here, we train the mixture of the exponential mixture model, which is a common model used in time series and stochastic process simulation.
\begin{figure}[hp]
    \begin{tikzpicture}
        \node[]{
    \includegraphics[width=0.6\linewidth]{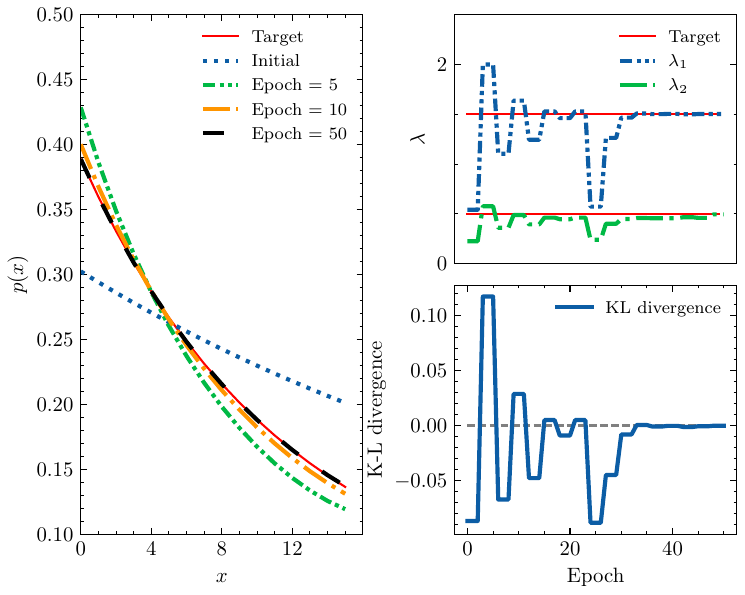}};
    \node[] at(-5.2, 4) {(a)};
    \node[] at(0.5, 4) {(b)};
    \node[] at(0.5, 0.1) {(c)};
    \end{tikzpicture}
    \centering
    \caption{
    {Learning Exponential Mixture Model based on SI-PQC.}
    {(a)}
    Quantum states in the learning procedure.
    {(b)}
    Estimated $\lambda$ extracted from the angles in the training procedure.
    {(c)}
    The K-L divergence to characterize the distance from the target state.
    }
    \label{fig:learning_exp}
\end{figure}

\let\oldbibitem\bibitem
\renewcommand{\bibitem}[2][]{
  \ifstrempty{#1}{\oldbibitem{SM_#2}}{\oldbibitem[#1]{SM_#2}}
}

\end{document}